\documentclass[a4paper,12pt]{article}
\usepackage[dvips]{color}
\usepackage{amsmath,amssymb,amsthm}
\usepackage[dvips]{graphicx}
\newtheorem{theorem}{Theorem}
\newtheorem{lemma}{Lemma}
\newtheorem{proposition}{Proposition}
\numberwithin{equation}{section}
\numberwithin{theorem}{section}
\numberwithin{lemma}{section}
\numberwithin{proposition}{section}
\begin{document}
\title{SCATTERING OF PARTICLES BOUNDED TO AN INFINITE PLANAR CURVE}
\author{J. Dittrich\\
Nuclear Physics Institute, Czech Academy of Sciences\\
CZ-250 68 \v{R}e\v{z}, Czech Republic\\
dittrich@ujf.cas.cz}
\date{}
\maketitle
\abstract{
Non-relativistic quantum particles bounded to a curve in ${\mathbb R}^2$ by attractive contact $\delta$-interaction are considered. The interval between the energy of the transversal
bound state and zero is shown to belong to the absolutely continuous spectrum,
with possible embedded eigenvalues. The existence of the wave operators is proved for the  mentioned energy interval using the Hamiltonians with the interaction supported by the straight lines as the free ones.
Their completeness is not proved.
The curve is assumed $C^3$-smooth, non-intersecting, unbounded, asymptotically approaching two different half-lines (non-parallel or parallel but excluding the "U-case"). Physically, the system can be considered as a model of long nanostructural channel.
}
\vspace{0.5cm}


\section{Introduction}

The contact interactions supported by curves or other low-dimensional structures are of interest for more than 20 years, with some renewed interest during the recent years.  Let us give only a sample of references
\cite{BEKS}-\cite{BHMP_arxiv_2019}.
Mostly bound states and compact curves or surfaces were studied.

We consider here the case of a planar unlimited curve supporting transversal attractive $\delta$-interaction of the strength $\alpha$. Except of the trivial case of the straight line, there exist isolated bound states of energy below $-\alpha^2/4$
while the interval $[-\alpha^2/4,\infty)$ belongs to the essential spectrum \cite{EI2001}. For the curves very closed to the straight line there exists just one discrete eigenvalue \cite{EK2015}. We are interested in the scattering of particles with energies $(-\alpha^2/4,0)$ along the curve.
This problem was studied in \cite{EK2005} for the case of curve different from a line in a compact set only. Here we consider $C^3$-smooth curve $\Gamma$ which is asymptotically approaching two half-lines at the infinity,
\begin{eqnarray*}
\Gamma(s)\thickapprox a_{\pm} + v_{\pm}s \quad {\rm for} \quad s \to \pm\infty \quad ,
\\
a_{\pm}, v_{\pm} \in {\mathbb R}^2 \quad , \quad |v_{\pm}|=1 \quad , \quad v_+ \not= -v_- \quad.
\end{eqnarray*}
The curve is assumed non-intersecting, parameterized by its length $s\in {\mathbb R}$ measured from a conventionally chosen point. The asymptotic conditions at $s\to \pm \infty$ more precisely specified below
are rather severe and probably not the optimal.

We show that the interval $(-\alpha^2/4,0)$ belongs to the absolutely continuous spectrum,
the existence of the embedded eigenvalues is not excluded. Further, we show the existence of the wave operators corresponding to the particles
in this energy range incoming or outgoing along the asymptotic half-lines. The completeness of the wave operators is not proved.

The proof of the absolute continuity of the spectrum follows the pattern of that for the Agmon-Kato-Kuroda theorem \cite{Agmon1975}. Essentially, it is based on the explicit verification of the
limiting absorption principle for the considered model. Recently, the limiting absorption principle for the Schr\"{o}dinger operators with contact interactions on compact surfaces has been considered in \cite{MPS2017}
and further generalized to the Dirac operators in \cite{BHMP_arxiv_2019}. As these papers, we use Krein-type resolvent formula and the weighted Lebesgue and Sobolev spaces. However, for the non-compact interaction support we cannot use theorems on compact embedding
but the assumptions on the curve asymptotics and some explicit estimates are sufficient for our results.

The existence of the wave operators is proved with the help of Kuroda-Birman theorem, i.e. essentially by proving that the full and free resolvent difference is a trace-class operator. Motivated by physics, we are interested
in the scattering of particles roughly speaking bounded near the curve, i.e. a spectral projector to the negative energies and also projectors to the direction of motion (momentum spectral projectors)
enter the definition of the wave operators and the assumptions of the Kuroda-Birman theorem. The presence of the projectors is substantially used together with the curve asymptotics at the technical level at the least.
This is the reason why our proofs do not give the completeness of the wave operators which has been proved for similar models with the compact interaction support (e.g., Section 2.4 of \cite{EK2005}, Corollary 4.13 of \cite{MPS2016}, Theorem 5.6 of \cite{BMN2017} and Theorem 4.1 of \cite{BHMP_arxiv_2019}).

The model is defined and the assumptions are formulated in the next section. Some properties of the model are reviewed in Sections \ref{curvesection}-\ref{essspectsection}. The relations
$(-\alpha^2/4,0) \subset \sigma_{ac}$ and $(-\alpha^2/4,0) \cap \sigma_{sc} = \emptyset$ are proved in Section \ref{acsection} with the result formulated in Theorem~\ref{spectralac}.
The existence of the wave operators is proved in Section \ref{wosection} with the result in Theorem \ref{woexistence}. Slight generalizations of the Pearson and Kuroda-Birman theorems used in the
proof are given in Appendix B. Section \ref{conclusionsection} contains brief concluding remarks.

\section{The model definition}
\label{modeldefsection}
We consider non-relativistic quantum mechanical particles moving in a plane with a contact interaction supported by a curve. Hilbert space of states of our system is the space $L^2({\mathbb R}^2)$.
The Hamiltonian is a self-adjoint operator $H$ associated with the sesquilinear form
\begin{equation}
\label{H_form}
q(f,g)=\int_{{\mathbb R}^2} \overline{\nabla f} \cdot \nabla g \, d^2 x - \alpha \int_{\mathbb R} \overline{f(\Gamma(s))} g(\Gamma(s)) \, ds \quad.
\end{equation}
The domain of the form is ${\mathcal D}(q) = H^{1,2}({\mathbb R}^2)$ and in the integral over curve we identify the value and the trace of functions on the (graph of) the curve for notational simplicity. The form $q$ is closed and below bounded \cite{BEKS} so $H$ is a well defined self-adjoint operator. We consider only the attractive case $\alpha > 0$.

We assume the following on the curve $\Gamma$.
\vspace{0.5cm}

\noindent
{\bf Assumption 1.} The curve $\Gamma: {\mathbb R} \to {\mathbb R}^2$ satisfies $\Gamma \in C^2$, $|\Gamma'(s)| = 1$ for $s\in{\mathbb R}$.
Let us write
\begin{eqnarray}
\label{nudef}
\Gamma(s) = a_\pm + v_\pm s + \nu_\pm(s) \quad,
\\
\label{phidef}
\varphi_-(s)=\sup_{t\leq s} |\nu_-'(t)| \quad , \quad \varphi_+(s)=\sup_{t\geq s} |\nu_+'(t)| \quad .
\end{eqnarray}
with some $a_-,a_+,v_-,v_+ \in {\mathbb R}^2$,
\begin{equation}
\label{vnorm}
|v_-| = |v_+| = 1 \quad
.
\end{equation}
There exist $R>0$ and $\delta>3$ such that
\begin{equation}
\label{phiLdelta}
\varphi_-\in L^2_\delta ((-\infty,-R)) \quad , \quad \varphi_+\in L^2_\delta ((R,\infty)) \quad .
\end{equation}
\vspace{0.5cm}

\noindent
Here $L^2_\delta$ are standard weighted Lebesgue spaces, e.g.
$$
L^2_\delta((R,\infty)) = \{ f: (R,\infty)\to {\mathbb C} \, | \, f \,{\rm measurable,} \, \int_R^\infty (1+s^2)^\delta \, |f(s)|^2  \, ds < \infty \}
$$
identifying the functions which are equal almost everywhere.

\vspace{0.5cm}
\noindent
{\bf Assumption 2.}
There exists $\rho \in (0,1]$ such that
\begin{equation}
\label{rhodef}
\rho |s-t| \leq |\Gamma(s)-\Gamma(t)|
\end{equation}
for every $s,t\in {\mathbb R}$.
\vspace{0.5cm}

\noindent
We shall see that under the other assumptions, the existence of $\rho$ is essentially equivalent to the non-intersecting of the curve.

\vspace{0.5cm}
\noindent
{\bf Assumption 3.} The curve $\Gamma$ has a uniformly bounded second derivative, i.e., $|\Gamma''| \in L^\infty({\mathbb R})$.
\vspace{0.5cm}

\noindent
The Assumption 3 means that the curve $\Gamma$ has a bounded extrinsic curvature
\begin{equation}
\label{curvature_def}
{\mathfrak K} = \Gamma_1' \Gamma_2'' - \Gamma_2' \Gamma_1''   \quad .
\end{equation}
It simplifies considerably mainly the study of the Hamiltonian domain but its necessity remains an open question. The following assumption we shall use for
a simple location of the essential spectrum but its necessity is also under the question.

\vspace{0.5 cm}
\noindent
{\bf Assumption 4.} The curve $\Gamma \in C^3(\mathbb R)$ and the curvature  together with its derivative vanish at the infinity in the sense that
\begin{equation}
\label{curv_der_decay}
\lim_{s\to \pm \infty} {\mathfrak K}(s)=0 \quad,\quad \lim_{s\to \pm \infty} {\mathfrak K}'(s)=0 \quad .
\end{equation}
\vspace{0.2cm}

\noindent
Notice that by the well known relations, (\ref{curv_der_decay}) is equivalent to
$$
\lim_{s\to\pm\infty} \nu_\pm''(s)=0 \quad,\quad \lim_{s\to\pm\infty} \nu_\pm'''(s)=0 \quad .
$$
\section{On the curve properties}
\label{curvesection}
We give some remarks on the assumed properties of the curve.
The all used assumptions are always mentioned in the text of propositions in this section as the whole set would be unnecessarily strong and as an alternative to the Assumption 2 is also given.

Evidently,
\begin{equation}
\rho |s-t| \leq |\Gamma(s)-\Gamma(t)| \leq |s-t|
\end{equation}
the first inequality being assumed and the second one following from the assumption $|\Gamma'(s)| = 1$ and saying only that a line is a geodesics in the Euclidean plane.
Equations (\ref{nudef}-\ref{phidef}) are just definition of the functions $\nu_\pm, \varphi_\pm$ and only assumption (\ref{phiLdelta}) gives them some contents.
\begin{proposition}
\label{curvelimits}
Under the Assumption 1, such $a_\pm$ can be chosen that
\begin{equation}
\label{nulim}
\lim_{s\to \pm \infty} \nu_\pm(s) = 0
\end{equation}
and then for any $1<\kappa<\delta + \frac{1}{2}$
(in particular for any $1<\kappa\leq 7/2$)
there exists $R_1>0$ such that
\begin{equation}
\label{nufall}
0\leq \varphi_\pm(\pm s)\leq s^{-\kappa} \quad, \quad |\nu'_\pm(\pm s)| \leq s^{-\kappa} \quad,\quad |\nu_\pm(\pm s)| \leq \frac{s^{1-\kappa}}{\kappa - 1}
\end{equation}
for $s\geq R_1$.
\end{proposition}
\begin{proof} Let us give the proof for $s>0$, the case $s<0$ being analogical. By (\ref{phidef}), $\varphi_+$ is non-increasing. Let us fix $\kappa\in (1, \delta+\frac{1}{2})$ and assume that
the first statement (\ref{nufall}) does not hold, i.e.
$$
(\forall R_1>0)(\exists s_1\geq R_1)(\varphi_+(s_1)>s_1^{-\kappa}) \quad.
$$
Now
\begin{eqnarray*}
\infty > \int_R^\infty (1+s^2)^\delta \, |\varphi_+(s)|^2 \, ds \geq |\varphi(s_1)|^2 \int_R^{s_1} (1+s^2)^\delta \, ds
\\
\geq \frac{s_1^{-2\kappa}}{2\delta+1} (s_1^{2\delta+1}-R^{2\delta+1})\quad .
\end{eqnarray*}
Here the limit $s_1\to\infty$ can be taken over a suitable sequence of $s_1$ values and a contradiction is obtained. So the first statement (\ref{nufall}) is proved and the second follows from
definition (\ref{phidef}).

For $s_2>s_1>R_1$ we now obtain
$$
|\nu_+(s_2)-\nu_+(s_1)| = \left|\int_{s_1}^{s_2} \nu'(s) \, ds\right| \leq \int_{s_1}^{s_2} s^{-\kappa} \, ds = \frac{1}{\kappa-1}(s_1^{1-\kappa}-s_2^{1-\kappa})
$$
and the existence of finite limit $\lim_{s\to \infty} \nu_+(s)$ follows. Absorbing its value into $a_+$ we may have (\ref{nulim}) without loss of generality.
The last inequality (\ref{nufall}) then follows by taking the limit $s_2\to \infty$ in the above estimate.
\end{proof}

The strong assumptions (\ref{phiLdelta}) were used in the proof. However, the true motivation for them is the proof of compactness of some operator below. We shall always assume (\ref{nulim})
in the following.
\begin{proposition}
\label{assump_equivalence}
Let $\Gamma \in C^2$, $|\Gamma'|=1$, (\ref{nudef}), (\ref{vnorm}), (\ref{nulim}) and
\begin{equation}
\label{nu_prime_lim}
\lim_{s\to\pm \infty} \nu_\pm '(s)=0
\end{equation}
hold. Then the following statements are equivalent.\\
i) There exists $\rho>0$ such that $|\Gamma(s) - \Gamma(t)| \geq \rho |s-t|$ for every $s,t\in {\mathbb R}$.\\
ii) The curve $\Gamma$ is non-intersecting (i.e. $\Gamma(s)\not=\Gamma(t)$ for $s\not= t$) and $v_+\not=-v_-$.
\end{proposition}
\begin{proof} Let us assume i). Then $\Gamma$ is non-intersecting. Assume that  $v_+=-v_-$. Then for large $s>0$
$$
\Gamma(s)-\Gamma(-s) = a_+ - a_- + \nu_+(s) - \nu_-(-s)
$$
remains bounded in contradiction with the relation $|\Gamma(s) - \Gamma(-s)| \geq 2\rho s$. So necessarily $v_+\not=-v_-$ and ii) holds.

Assume now ii). There exists $s_1>0$ such that $|\nu_+'(s)|<\frac{1}{2}$ for $s>s_1$. For $s,t>s_1$ now
$$
|\Gamma(s)-\Gamma(t)|=|v_+(s-t) +\nu_+(s)-\nu_+(t)|\geq \frac{1}{2} |s-t| \quad.
$$
Enlarging $s_1$ if necessary, the same holds for $s,t<-s_1$.

Let us consider case $s>0$, $t<0$ now. Then $|s-t|=s-t=s+|t|$ and
$$
\xi = \frac{s+t}{s-t} \in [-1,1] \quad.
$$
$$
F(s,t)=\left|\frac{v_+s-v_-t}{s-t}\right|=\left|\frac{1}{2}(v_++v_-)+\frac{1}{2}(v_+-v_-)\xi\right|
$$
takes its minimum $c(v_-,v_+)$ at some $\xi\in [-1,1]$. If the minimum would be zero then $v_+(1+\xi)=-v_-(1-\xi)$. Taking into account that $|v_+|=|v_-|$ this gives
$\xi=0$ and minimum value $\frac{1}{2}|v_++v_-|=0$. However, this contradicts the assumption $v_+\not=-v_-$. So
$$
|v_+s-v_-t| \geq c(v_-,v_+) |s-t|, \quad c(v_-,v_+) >0
$$
for $s>0$, $t<0$.

There exist  $s_2 \geq c(v_-,v_+)^{-1}(|a_+ - a_-| + 2)$ such that for $s>s_2$, $t<-s_2$ is $|\nu_+(s)|<1$, $|\nu_-(t)|<1$.
For $s>s_2$, $t<-s_2$ now
\begin{eqnarray*}
|\Gamma(s)-\Gamma(t)| \geq |v_+s-v_-t| - |a_+ - a_-| - |\nu_+(s)| - |\nu_-(t)|
\\
\geq c(v_-,v_+) |s-t| - (|a_+ - a_-| + 2)
\\
 \geq \frac{1}{2} c(v_-,v_+) |s-t| + c(v_-,v_+) s_2 - (|a_+ - a_-| + 2)
\\
\geq \frac{1}{2} c(v_-,v_+) |s-t| \quad.
\end{eqnarray*}
By the symmetry, the same holds for $s<-s_2$, $t>s_2$.

There exists $s_3\geq \max(s_1, s_2)$, $s_3>2(|a_+| +1 + |\Gamma(0)|$) such that $|\nu_+(s)|<1$ for $s>s_3$. For $s>2 s_3$ and $|t|<s_3$ now
\begin{eqnarray*}
|\Gamma(s)-\Gamma(t)|\geq |\Gamma(s)-\Gamma(0)| - |\Gamma(t)-\Gamma(0)|
\\
\geq |a_+ + v_+ s +\nu_+(s) - \Gamma(0)| -s_3
\\
\geq s - |a_+| - 1 - |\Gamma(0)| - s_3 \geq \frac{1}{4} s > \frac{1}{6} |s-t| \,.
\end{eqnarray*}
Enlarging $s_3$ if necessary, we can obtain such estimate for every $|s|>2 s_3 \, , |t|<s_3$ and every $|s|<s_3 \, , |t|> 2 s_3$.

As the last case, consider function
$$
G(s,t) = \left|\frac{\Gamma(s) - \Gamma(t)}{s-t}\right|
$$
on the compact $[-2s_3,2s_3]\times [-2s_3,2s_3]$. As $\Gamma\in C^2$, G can be continuously extended to the values of $s=t$,
$$
G(s,s)=|\Gamma'(s)|=1 \quad.
$$
G takes its minimum $\rho_1$ which is nonzero as $\Gamma$ is non-intersecting by the assumption. As a result,
\begin{equation}
\label{rho1}
|\Gamma(s)-\Gamma(t)|\geq \rho_1 |s-t| \quad , \quad \rho_1>0
\end{equation}
for $|s|\leq 2 s_3$, $|t| \leq 2 s_3$.

Putting $\rho=\min(\frac{1}{2},\frac{1}{2}c(v_-,v_+),\frac{1}{6},\rho_1)$, the statement i) is proved.
\end{proof}
By Proposition \ref{curvelimits}, the assumption (\ref{nu_prime_lim}) is weaker then the corresponding part of Assumption 1. Let us repeat some properties
of the curve $\Gamma$ without distinguishing assumptions and their consequences for the summary.
\begin{eqnarray}
\Gamma(s) = a_\pm + v_\pm s +\nu_\pm(s) \quad , \quad |v_\pm|=1 \quad ,\quad v_ + \not= -v_- \quad ,
\\
\lim_{s\to\pm\infty} \nu_\pm(s) = 0 \quad , \quad \lim_{s\to\pm\infty} \nu_\pm'(s) = 0 \quad ,
\\
\rho |s-t| \leq |\Gamma(s) - \Gamma(t)| \leq |s-t| \quad , \quad 0<\rho\leq 1 \quad , \quad s,t\in{\mathbb R} \quad .
\end{eqnarray}
\vspace{1mm}
\begin{proposition}
Under the Assumptions 1 and 2,
the curve $\Gamma$ intersects the plain ${\mathbb R}^2$ into two disjoint regions $\Omega_+$, $\Omega_-$, i.e.,
$$
{\mathbb R}^2=\Omega_+ \cup \Gamma({\mathbb R}) \cup \Omega_-
$$
where the union is disjoint.
\end{proposition}
\begin{proof} Under our assumptions, $\Gamma$ is a Jordan curve in the completed plane ${\mathbb R}^{2*}\approx S^2$ and so intersects it into two disjoint
regions by the Jordan theorem (e.g. Theorem 108 in \cite{Cerny}).
\end{proof}
\section{Operator domain}
\label{opdomainsection}
Although the Hamiltonian is completely defined by the form (\ref{H_form}) it is also useful to know the domain of the corresponding operator $H$
itself and its action,
\begin{eqnarray}
\nonumber
{\mathcal D}(H)=
\\
\nonumber
\{\psi\in H^{1,2}({\mathbb R}^2) \cap  H^{2,2}_{loc}({\mathbb R}^2\setminus \Gamma({\mathbb R})) \,|\,\partial_n\psi_{|\Gamma_+}-\partial_n\psi_{|\Gamma_-}=\alpha \psi_{|\Gamma},\\
\nonumber
\Delta\psi\in L^2({\mathbb R}^2) \}
\quad ,
\\
\label{Laplacian}
H\psi = -\Delta\psi \quad {\rm for} \quad \psi \in {\mathcal D}(H) \quad
\end{eqnarray}
where $\partial_n\psi_{|\Gamma_\pm}$ are the traces of the normal derivative of $\psi$ at the graph of $\Gamma$ from $\Omega_\pm$ respectively.
The unit vector $n$ normal to $\Gamma$ is oriented from $\Omega_+$ into $\Omega_-$ and $\psi_{|\Gamma}$ is the trace of $\psi$ (the same from the both sides for
$\psi \in {\mathcal D}(q) \supset {\mathcal D}(H)$). The Laplacian in (\ref{Laplacian}) is considered in the distributional sense on
$\Omega_+ \cup \Omega_-$. The space $H^{2,2}_{loc}({\mathbb R}^2\setminus \Gamma({\mathbb R}))$ should be understood here as the space of all functions which are
in
$$
H^{2,2}(U\setminus \Gamma({\mathbb R}))=H^{2,2}(U\cap\Omega_+) \oplus H^{2,2}(U\cap\Omega_-)
$$
for every bounded open $U\subset {\mathbb R}^2$.
Evidently, the interchange of conventionally chosen $\Omega_+$ and $\Omega_-$ has no effect as it leads to the change
of orientation of $n$.

Assuming further that $\Gamma '' \in L^\infty({\mathbb R})$ (Assumption 3),
\begin{eqnarray}
\nonumber
{\mathcal D}(H)=
\\
\label{H-domain}
\{\psi\in H^{1,2}({\mathbb R}^2) \cap  H^{2,2}({\mathbb R}^2\setminus \Gamma({\mathbb R})) \,|\,\partial_n\psi_{|\Gamma_+}-\partial_n\psi_{|\Gamma_-}=\alpha \psi_{|\Gamma} \}
\quad .
\end{eqnarray}

The form of ${\mathcal D}(H)$ in (\ref{H-domain}) is given in Theorem 8.4 of \cite{BLLR2017}.
A relatively short proof in Sect. 2.1 of \cite{Galk-Smith} for the compact curve extends to the proof of (\ref{Laplacian}) and (\ref{H-domain}) in our case of infinite asymptotically flat curve. The only points needing special attention
are the trace formula
$$
\|u\|_{L^2(\Gamma(R))} \leq C \|u\|_{L^2({\mathbb R}^2)}^\frac{1}{2} \|u\|_{H^{1,2}({\mathbb R}^2)}^\frac{1}{2}
$$
which can be proved integrating the derivative along the intervals of uniform length orthogonal to the asymptotic of our curve, instead to the curve itself as in the standard proof.
For (\ref{H-domain}), we can cover the outgoing parts of the curve by a countable set of finitely intersecting bounded rectangular neighborhoods of the same size and so the same
constants appear in the estimates of local $H^{2,2}$-norms by the local $H^{1,2}_\Delta$-norms. Then the local estimates extend to the global one.
\section{The resolvent}
\label{resolventsection}
The resolvent for the Hamiltonian of our model is known from \cite{BEKS}. Let us define operators\footnote{Notice that in Eq. (2.2) of \cite{BEKS},
$ik$ should be changed to $-ik$ for $d\geq 2$.}
\begin{eqnarray}
\nonumber
R_{dx\,dx}(k) : L^2({\mathbb R}^2) \to L^2({\mathbb R}^2) \quad,
\\
\label{Rxx}
(R_{dx\,dx}(k)\psi)(x)=\frac{1}{2\pi}\int_{{\mathbb R}^2} K_0(-ik|x-y|)\psi(y)\, dy \quad,
\\
\nonumber
R_{dx\,m}(k) : L^2({\mathbb R}) \to L^2({\mathbb R}^2) \quad,
\\
\label{Rxm}
(R_{dx\,m}(k)f)(x)=\frac{1}{2\pi}\int_{{\mathbb R}} K_0(-ik|x-\Gamma(s)|)f(s)\, ds \quad,
\\
\nonumber
R_{m\,dx}(k) : L^2({\mathbb R}^2) \to L^2({\mathbb R}) \quad,
\\
\label{Rmx}
(R_{m\,dx}(k)\psi)(s)=\frac{1}{2\pi}\int_{{\mathbb R}^2} K_0(-ik|\Gamma(s)-y|)\psi(y)\, dy \quad,
\\
\nonumber
R_{m\,m}(k) : L^2({\mathbb R}) \to L^2({\mathbb R}) \quad,
\\
\label{Rmm}
(R_{m\,m}(k)f)(s)=\frac{1}{2\pi}\int_{{\mathbb R}} K_0(-ik|\Gamma(s)-\Gamma(t)|)f(t)\, dt \quad.
\end{eqnarray}
Here $k\in{\mathbb C}$, $\Im k >0$ (with possible analytic prolongations), $K_0$ is the Macdonald function, the space $L^2({\mathbb R})\approx L^2(\Gamma({\mathbb R}), ds)$ is understood as the space of functions defined on the curve parameterized by the length.

In this notation, the resolvent (see Corollary 2.1 in \cite{BEKS} with a slightly different notation)
\begin{equation}
\label{full_resolvent}
(H-k^2)^{-1} = R_{dx\, dx}(k) +\alpha R_{dx\, m}(k) (I-\alpha R_{m\, m}(k))^{-1} R_{m\, dx}(k) \quad.
\end{equation}

For the further references, let us remind here some properties and estimates of the
Macdonald function $K_\nu$ (also called as modified Bessel function of the third kind; $\nu=0,1,2,\dots$ is sufficient for us) which we shall extensively use.
It holds
\begin{equation}
\label{Knureal}
|K_\nu(z)|\leq K_\nu(\Re z)
\end{equation}
for $\nu\geq 0$, $\Re z >0$
and $K_\nu$ is decreasing in $z>0$. Further
\begin{eqnarray}
\label{K0bound}
|K_0(z)|\leq c_0 (1+|\log|z||) e^{-\Re z} \quad , \quad |K_n(z)| \leq c_n |z|^{-n} e^{-\frac{3}{4}\Re z} \quad,
\\
\label{K0derivative}
|K_0^{(k)}(z)| \leq c_k' |z|^{-k} e^{-\frac{1}{2}\Re z}
\end{eqnarray}
for $|z|>0$, $|\arg z| \leq \frac{\pi}{2} -\varepsilon$ with constants $c_0,\, \, c_n, \, c_k'>0$ depending on $\varepsilon\in (0,\frac{\pi}{2})$ ($k, n=1,2,\dots$).
These properties follow from the well known integral representation (equation (16) in par. 7.3 of \cite{Bat-Erd}) giving estimates for large $|z|$ and behavior near $z=0$.
\section{The essential and discrete spectrum}
\label{essspectsection}
By Proposition 5.1 of \cite{EI2001} we would have
\begin{equation}
\label{essential-spectrum}
\sigma_{ess}(H)=[-\frac{\alpha^2}{4},\infty)
\end{equation}
if we would clarify the two following points.

The validity of the assumptions of \cite{EI2001} under our assumptions can be checked, our assumptions are much stronger in fact. Details are given in the Appendix A.

Strictly speaking, only the inclusion $\sigma_{ess}(H)\subset [-\frac{\alpha^2}{4},\infty)$ is proved in \cite{EI2001} while the proof of the opposite inclusion is incomplete there (see the Remark below).
One would conjecture, that the proof could be completed but instead of desirable clarification of this point, we give a straightforward proof of the relation (\ref{essential-spectrum}) under the stronger assumptions.
\begin{proposition}
\label{essspectr}
Under the Assumptions 1-4,
\begin{equation}
\label{essential-spectrum-prop}
\sigma_{ess}(H)=[-\frac{\alpha^2}{4},\infty)   \quad .
\end{equation}
\end{proposition}
\begin{proof} By the corresponding part of the proof of Proposition 5.1 in \cite{EI2001}, $\sigma_{ess}(H) \cap (-\infty, -\frac{\alpha^2}{4}) = \emptyset$. We shall construct a Weyl sequence for the points
$-\frac{\alpha^2}{4} + k^2$, $k\in {\mathbb R}$ now.

Let us first introduce local coordinates $s$ (curve length) and $t$ (roughly distance from the curve) in a neighborhood of the curve (cf. \cite{DE1995}). Let $n(s)=(-\Gamma_2'(s),\Gamma_1'(s))$
be the unit normal vector to the curve and let us write
$$
(x_1,x_2)=\Gamma(s)+t n(s).
$$
The local coordinates $(s,t)$ may be used in a region where
$s\in J$ (an open interval in ${\mathbb R}$), $t\in (-A,A)$, $A>0$ such that $A {\mathfrak K}_0<1$,  ${\mathfrak K}_0= \|{\mathfrak K}\|_{L^\infty(J)}$, the length of $J$ does not exceed
${\mathfrak K}_0^{-1}-A$ and $\Gamma({\mathbb R}\setminus J)$ does not enter the region.
In particular, the Jacobian $1-t {\mathfrak K}(s) \not= 0$ there.
Given any $s_0>0$ sufficiently large, $L_1>0$, $L_2>0$,  the curvilinear rectangle $s_0-1\leq s \leq s_0+L_1+1$,
$-L_2-1\leq t\leq L_2+1$ can be placed inside this region (see Assumption 4 and  Propositions \ref{curvelimits} and \ref{assump_equivalence}).  The Laplacian in the coordinates $s$, $t$ reads
$$
-\Delta=-(1-t{\mathfrak K})^{-2}\frac{\partial^2}{\partial s^2} -t {\mathfrak K'}(1 -t{\mathfrak K})^{-3}\frac{\partial}{\partial s} -\frac{\partial^2}{\partial t^2}
+ {\mathfrak K} (1-t{\mathfrak K})^{-1} \frac{\partial}{\partial t}  \quad.
$$

Let $\vartheta \in C^\infty({\mathbb R})$ be such that $0 \leq \vartheta \leq 1$, $\vartheta(x)=0$ for $x\leq 0$, and $\vartheta(x)=1$ for $x\geq 1$.
Let us define $\vartheta_1(s)= \vartheta(s-s_0+1)$ for $s\leq s_0 + L_1$, $\vartheta_1(s)=\vartheta(s_0+L_1+1-s)$ for $s \geq s_0 +L_1$.
Similarly,  $\vartheta_2(t)=\vartheta(t+L_2+1)$ for $t\leq L_2$, $\vartheta_2(t)=\vartheta(L_2+1-t)$ for $t\geq L_2$.
Let
\begin{equation}
\label{testfunction}
\psi(s,t)=e^{-\frac{\alpha}{2} |t|}e^{i k s} \vartheta_1(s) \vartheta_2(t)   \quad.
\end{equation}
Then $\psi \in {\mathcal D}(H)$ (\ref{H-domain})  and
\begin{equation}
\label{normpsibound}
\|\psi\|^2=\int |\psi(s,t)|^2 |1-t{\mathfrak K}(s)| \, ds \, dt  \geq \alpha^{-1}L_1(1-e^{-\alpha L_2})
\end{equation}
if
\begin{equation}
\label{curvbound}
|{\mathfrak K}(s)| \leq \frac{1}{2(L_2+1)} \quad {\rm for} \quad s\geq s_0 -1  \quad.
\end{equation}
Direct calculations give
\begin{equation}
(-\Delta + \frac{\alpha^2}{4}-k^2)\psi=\psi_1+\psi_2+\psi_3+\psi_4
\end{equation}
where denoting $d(s,t)=(1-t{\mathfrak K}(s))^{-1}$
\begin{eqnarray*}
\psi_1=k^2(d^2-1)\psi    \quad,
\\
\psi_2=-(i k t {\mathfrak K}' d^3 + \frac{1}{2} \alpha {\mathfrak K} d {\rm\, sgn}(t))\psi        \quad,
\\
\psi_3=\left( (\alpha {\rm\, sgn}(t) +{\mathfrak K} d)\vartheta_2'(t) - \vartheta_2''(t) \right) e^{-\frac{\alpha}{2} |t|}e^{i k s} \vartheta_1(s)    \quad,
\\
\psi_4=-d^2\left((2 i k + t {\mathfrak K}' d) \vartheta_1'(s) + \vartheta_1''(s) \right) e^{-\frac{\alpha}{2} |t|}e^{i k s} \vartheta_2(t)    \quad.
\end{eqnarray*}
We denote
$$
{\mathfrak K}_0=\sup_{s\geq s_0-1}|{\mathfrak K}(s)| \quad , \quad {\mathfrak K}_1=\sup_{s\geq s_0-1}|{\mathfrak K}'(s)| \quad. $$
Assuming (\ref{curvbound}) and denoting c a suitable constant independent of $s_0$, $L_1$ and $L_2$ (but dependent on $k$, $\alpha > 0$ and
$\|{\mathfrak K}\|_{L^\infty({\mathbb R})}$) we can estimate
\begin{eqnarray*}
\|\psi_1\|^2 \leq c {\mathfrak K}_0^2 (L_1+2) \quad, \\
\|\psi_2\|^2\leq c ({\mathfrak K}_1^2 + {\mathfrak K}_0^2)(L_1+2) \quad,
\\
\|\psi_3\|^2\leq c e^{-\alpha L_2} (L_1+2) \quad,
\\
\|\psi_4\|^2 \leq c (1+{\mathfrak K}_1^2)
\end{eqnarray*}
which together with (\ref{normpsibound}) gives (c is another suitable constant)
$$
\frac{\|(H+\frac{\alpha^2}{4}-k^2)\psi\|}{\|\psi\|} \leq
c  \frac{({\mathfrak K}_0 + {\mathfrak K}_1 +e^{-\frac{\alpha}{2}L_2})\sqrt{L_1+2} +(1+{\mathfrak K}_1)}{\sqrt{L_1 (1-e^{-\alpha L_2})}}
\quad .
$$
This can be arbitrarily small by the choice of $L_1$, $L_2$ and $s_0$, keeping the validity of (\ref{curvbound}). So the Weyl sequence
can be constructed and $-\frac{\alpha^2}{4} +k^2 \in \sigma_{ess}(H)$.
\end{proof}

\noindent
{\bf Remark.}
The proof of Proposition 5.1 in \cite{EI2001} is incomplete as it uses implication $1\in \sigma(\alpha R_{m\, m}(ik)) \Rightarrow -k^2 \in \sigma(H)$
for $k>0$ referring to part (ii) of Proposition 2.1 which is proved as Corollary 2.1 of \cite{BEKS} and which contains opposite implication for resolvent sets,
$1\not\in \sigma(\alpha R_{m\, m}(ik)) \Rightarrow -k^2 \not\in \sigma(H)$, only.

The equivalence like $1\in \sigma(\alpha R_{m\, m}(ik)) \Leftrightarrow -k^2 \in \sigma(H)$ would follow from Theorem 1.29 in \cite{BGP2008}.
However, \cite{BGP2008} consider the case
of boundary triples which are insufficient for the partial differential operators where quasi-boundary triples should be used (cf. discussions in the Introductions
of \cite{BLL2013, BL2007}).  The proof of a similar Theorem A in \cite{BLL2012} on the essential spectrum rely on compact embedding of $H^1(\Gamma({\mathbb R}))$ into $L^2(\Gamma({\mathbb R}))$ which does not hold
for the infinite curve.
\\

Below the essential spectra, eigenvalues discrete in $(-\infty,-\alpha^2/4)$ may appear. If the curve is not just a straight line, there is at most one such eigenvalue \cite{EI2001}. The whole spectrum is bounded from below as the form (\ref{H_form}) is. The accumulation of the eigenvalues at $-\alpha^2/4$ cannot be probably excluded in general but for the our case of asymptotically flat curve one would conjecture finite number of eigenvalues only, although the proof remains an open question. For the curve in a sense close to the straight line, there exists just one simple discrete eigenvalue. This is a very slight generalization of Theorem 2.2 from \cite{EK2015} which we give after the introduction of some notations.

Let $\Gamma$ be the curve satisfying Assumptions 1-4, put the origin of coordinates into $\Gamma(0)=0$ for the simplicity. We define a relative angle of the tangent $\phi$,
$$
\Gamma'(s)=(\cos \phi(s,0), \sin \phi(s,0)) \quad,\quad \phi(s,t)=\phi(s,0)-\phi(t,0) \quad.
$$
The angle $\phi(\cdot,0)$ can be defined as a bounded function from $C^2({\mathbb R})$ under our assumptions. Let us define a deformed curve
$$
\Gamma^{(\beta)}(s)= \int_0^s (\cos \beta \phi(t,0), \sin \beta \phi(t,0)) \, dt \quad,\quad \beta \in [0,1] \quad,
$$
in other words a homotopy of the curve $\Gamma$ to the straight line. For $\beta$ small enough, $\Gamma^{(\beta)}$ satisfies Assumptions 1-4 and let us denote as $H^{(\beta)}$ the corresponding Hamiltonian.
Let us denote
$$
{\mathcal A}(s,t) = -\frac{\alpha^4}{32\pi}K_1(\tfrac{\alpha}{2} |s-t|) \left( |s-t|^{-1}\left(\int_t^s \phi(u,t) \, du\right)^2 - \left|\int_t^s \phi(u,t)^2 \, du \right| \right) \quad.
$$
Under our assumptions, the convergence of integral
$$
\int_{{\mathbb R}^2} {\mathcal A}(s,t) \, ds \, dt < \infty
$$
can be seen.
\begin{proposition}
\label{eigenvasympt}
Let a curve $\Gamma$ satisfies Assumptions 1-4 and is not a straight line. Then there exists $\beta_0>0$ such that for $0<\beta<\beta_0$ the Hamiltonian $H^{(\beta)}$ has a unique simple
eigenvalue $\lambda(H^{(\beta)})$ in $(-\infty, -\frac{\alpha^2}{4})$ which admits the asymptotic expansion
$$
\lambda(H^{(\beta)}) = -\frac{\alpha^2}{4} - \left(\int_{{\mathbb R}^2} {\mathcal A}(s,t) \, ds \, dt \right)^2 \beta^4 + o(\beta^4)
$$
for $\beta \to 0^+$.
\end{proposition}
\begin{proof}
The proof copies step by step the proof of Theorem 2.2 in \cite{EK2015} so we do not repeat details here. We need an asymptotic estimate of the angle $\phi$. Under our assumptions,
finite
$$
\lim_{s\to \pm \infty} \phi(s,0) = \phi_{\pm\infty} \quad,\quad {\rm where} \quad v_{\pm} = (\cos \phi_{\pm \infty}, \sin \phi_{\pm \infty}) \quad,
$$
exists. For $|s|$ large enough we then have
$$
|\phi(s,0) - \phi_{\pm \infty}| \leq \pi | \sin \tfrac{1}{2}(\phi(s,0) - \phi_{\pm \infty}) | = \frac{\pi}{2} |\nu_\pm' (s)| \leq \frac{\pi}{2} |s|^{-7/2}
$$
by Proposition \ref{curvelimits}. For the auxiliary regularization, \cite{EK2015} uses an exponential function $V_{-\alpha}(s) = \exp(-\alpha |s|/4)$ which decreases too quickly for our case
but we use $W(s)=(1+s^2)^{-1}$ instead. With these changes, the proof from \cite{EK2015} can be repeated, including estimates of the Hilbert-Schmidt norms for some operators.
\end{proof}
\section{Spectral absolute continuity}
\label{acsection}
We shall prove that $(-\frac{\alpha^2}{4},0) \subset \sigma_{ac}(H)$ and that the singular continuous spectrum does not intersect this interval,
our first main result, in this section. The existence of the countable number of eigenvalues in the interval $(-\frac{\alpha^2}{4},0)$
is not excluded. The proof is based on the well known
criterion given by Theorem XIII.20 of \cite{RSIV} related to the so called limiting absorbtion principle, using formula (\ref{full_resolvent}), comparison with the operators
related to the case of straight line, and the ideas used in the proof of Agmon-Kato-Kuroda theorem (Theorem XIII.33 of \cite{RSIV}, see also \cite{Agmon1975}).
Only Assumptions 1 and 2 are needed for the proofs in this section except of the inclusion of interval $(-\frac{\alpha^2}{4},0)$ into the essential spectrum.
\emph{Assumptions 1 and 2 are always supposed in this section without repeating that in the statements.}

We consider the argument of the resolvent
\begin{eqnarray*}
k^2=\lambda + i \varepsilon \quad , \quad \lambda \in [\lambda_1,\lambda_2] \subset (-\frac{\alpha^2}{4},0) \quad , \quad 0<\varepsilon<\varepsilon_0 \quad ,
\\
k=k_1 +i k_2  \quad , \quad k_1\in {\mathbb R} \quad,\quad k_2 \geq 0
\end{eqnarray*}
with fixed bounds $\lambda_1 < \lambda_2 < 0$, $\varepsilon_0 > 0$.  These relations imply
\begin{equation}
\label{kbounds}
k_1\in (0,b] \quad,\quad k_2 \in [k_{20},k_{21}]
\end{equation}
with finite $b$, $0<k_{20}<k_{21}$, more precisely
\begin{eqnarray*}
b=\left( \frac{1}{2} (\lambda_2 +\sqrt{\lambda_2^2 + \varepsilon_0^2}) \right)^\frac{1}{2}  \leq \frac{\varepsilon_0}{2\sqrt{-\lambda_2}} \quad,
\\
\quad k_{20}=\sqrt{-\lambda_2} \quad, \quad
k_{21}=\left( \frac{1}{2} (-\lambda_1 +\sqrt{\lambda_1^2 + \varepsilon_0^2}) \right)^\frac{1}{2} \quad.
\end{eqnarray*}
\emph{The parameter $k$ is always supposed to satisfy these relations if nothing else is said in this section, without repeating that in the statements. }

Notice that we may allow also negative values of $\varepsilon$, i.e. $\varepsilon \in (-\varepsilon_0,\varepsilon_0)$ and (\ref{kbounds}) then changes to
\begin{equation}
\label{kboundsb}
k_1\in [-b,b] \quad,\quad k_2 \in [k_{20},k_{21}] \quad.
\end{equation}
The important positive lower bound on $k_2$ remains which allows an analytic prolongation of some functions below the interval $[\lambda_1,\lambda_2]$
of the real axis in the following.

The notation of operators introduced in Section \ref{resolventsection} is used, sometimes using the same symbol for the operators restricted or extended to other
spaces than just $L^2$. We also denote
$$
w(s)=(1+s^2)^\frac{1}{2}  \quad.
$$
\begin{lemma}
\label{Rmm_compact}
Let $k$ satisfies (\ref{kboundsb}). Then
for every $\delta \geq 0$ and $\eta<-\frac{1}{2}$, $R_{mm}(k)$ is a compact operator $L^2_\delta({\mathbb R}) \to L^2_\eta({\mathbb R})$.
\end{lemma}
\begin{proof} We first verify that the operator
$ w^\eta R_{mm}(k)w^{-\delta}: L^2({\mathbb R}) \to L^2({\mathbb R})
$
defined as
$$
(w^\eta R_{mm}(k)w^{-\delta}\, f)(s) = \frac{1}{2\pi}\int_{{\mathbb R}} w(s)^\eta K_0(-ik|\Gamma(s)-\Gamma(t)|) w(t)^{-\delta} f(t)\, dt
$$
is Hilbert-Schmidt and therefore compact. Inequalities from Section \ref{resolventsection}, (\ref{rhodef}) and (\ref{kbounds}) give
$$
|K_0(-i k |\Gamma(s)-\Gamma(t)|)| \leq  K_0(k_2 |\Gamma(s)-\Gamma(t)|) \leq K_0(k_{20} \rho  |s-t|)
$$
and we estimate the integral of the kernel square
\begin{eqnarray*}
\int_{\mathbb R} \int_{\mathbb R} w(s)^{2\eta} |K_0(-ik|\Gamma(s)-\Gamma(t)|)|^2 w(t)^{-2\delta} \, ds \, dt  \leq
\\
\int_{\mathbb R} \int_{\mathbb R} w(s)^{2\eta} K_0(k_{20} \rho  |s-t|)^2 \, ds \, dt =
\\
\left(\int_{\mathbb R} w(s)^{2\eta} \, ds\right) \left(\int_{\mathbb R} K_0(k_{20} \rho  |u|)^2 \, du \right)  <\infty
\end{eqnarray*}
so the mentioned operator is really Hilbert-Schmidt. Taking into account the isomorphisms of the spaces $L^2_{\pm\delta}({\mathbb R})$ and $L^2({\mathbb R})$
realized as multiplications by suitable powers of $w$, the compactness of $R_{mm}(k): L^2_\delta \to L^2_\eta$ follows.
\end{proof}
\begin{lemma}
Let $\delta \geq 0$, $\eta< -\frac{1}{2}$, $\lambda_1<\lambda_2<0$ and $\varepsilon_0>0$. Then
the operator function  $z \mapsto R_{mm}(\sqrt{z})$, $R_{mm}(\sqrt{z}): L^2_\delta({\mathbb R}) \to L^2_\eta({\mathbb R})$, with the square root chosen so that $\Im \sqrt{z}>0$,
is holomorphic in the operator norm  in the region $(\lambda_1,\lambda_2)+i(-\varepsilon_0,\varepsilon_0)$.
\end{lemma}
\begin{proof} Function $k=\sqrt{z}$ is holomorphic with the values inside (\ref{kboundsb}). So it is sufficient to prove that $R_{mm}(k)$ is holomorphic
in $k$ inside (\ref{kboundsb}).
Using relation $K_0'=-K_1$
and estimates (\ref{Knureal}), (\ref{K0derivative}), (\ref{rhodef}) we obtain
\begin{eqnarray*} \bigg| \frac{K_0(-i k'|\Gamma(s)-\Gamma(t)|) - K_0(-i k|\Gamma(s)-\Gamma(t)|)}{k'-k} -
\\
i |\Gamma(s)-\Gamma(t)| K_1(-i k|\Gamma(s)-\Gamma(t)|)\bigg|  \leq c |k'-k| e^{-\frac{1}{2} k_{20}\rho |s-t|}
\end{eqnarray*}
with a suitable constant c. Similarly as in the proof of Lemma \ref{Rmm_compact} we show that the derivative of $w^\eta R_{mm}(k)w^{-\delta}$ exists in the
Hilbert-Schmidt norm, and therefore also in $L^2$ norm and then the derivative of $R_{mm}(k)$ exists in the norm of $L^2_\delta({\mathbb R}) \to L^2_\eta({\mathbb R})$
operators.
\end{proof}

Let us define an auxiliary operator $R_0(k): L^2({\mathbb R}) \to L^2({\mathbb R})$ as
\begin{equation}
\label{R0def}
(R_0(k) f)(s)=\frac{1}{2\pi}\int_{\mathbb R} K_0(-ik|s-t|) f(t)\, dt
\end{equation}
using again the same symbol for its restrictions or extensions and assuming $\Im k>0$. This would be the operator $R_{mm}(k)$ in the case of straight line $\Gamma$. However, the
comparison of the two operators must not be confused with comparison of the say scattering along $\Gamma$ and along the line. Our space $L^2({\mathbb R},ds)$
still represents a space of functions defined on $\Gamma({\mathbb R})$.

The Fourier transformation ${\mathcal F}$ maps the integral operator $R_0(k)$ to the multiplication operator in $L^2({\mathbb R}, dp)$,
$$
{\mathcal F} R_0(k) {\mathcal F}^{-1} = \frac{1}{2\sqrt{p^2-k^2}}    \quad.
$$
This can be easily verified with the help of integral representation (7.3.15) from \cite{Bat-Erd} for $K_0$.
An operator similar to the one appearing in (\ref{full_resolvent}) can be now written as
\begin{equation}
\label{A_def}
{\mathcal F}(I-\alpha R_0(k))^{-1}{\mathcal F}^{-1} = A(k,p) = \frac{2\sqrt{p^2-k^2}}{2\sqrt{p^2-k^2} -\alpha}
\end{equation}
for $\Re k^2<0$ and $\Im k^2>0$. The last two formulae were already used in (2.7) and (2.8) of \cite{EK2005}.
\begin{lemma}
For $\Re k^2<0$ and $\Im k^2>0$, $A(k,p)$ defined in (\ref{A_def}) is a bounded operator in $L^2({\mathbb R})$ as well as in $H^{1,2}({\mathbb R})$.
For given $\varepsilon_1>0$, the operator is uniformly bounded for $\Im k^2 > \varepsilon_1$. The same holds for its derivatives with respect to $k^2$.
\end{lemma}
\begin{proof} For a fixed $k$ satisfying the assumptions, (\ref{A_def}) is a bounded continuous function of $p$ so it gives a
bounded operator in $L^2$. The same holds for its derivative with respect to $p$ and so the operator is bounded also in $H^{1,2}$. Uniform bounds for
$\Im k^2 > \varepsilon_1$ are also seen. The derivatives of $A(k,p)$ with respect to $k^2$ have also these properties.
\end{proof}
\begin{lemma}
\label{Aholomorph}
Let $-\frac{\alpha^2}{4}<\lambda_1<\lambda_2<0$, $\Re z \in (\lambda_1,\lambda_2)$, $\Im z \not= 0$ and
Z(z) be the operator of multiplication by $A(\sqrt{z},p)$ in $H^{1,2}({\mathbb R})$ or $L^2({\mathbb R})$. Then $Z$ is an operator function holomorphic in the considered range of $z$
weakly, strongly as well as with respect to the norms of ${\mathcal L}(H^{1,2}({\mathbb R}))$ and ${\mathcal L}(L^2({\mathbb R}))$.
\end{lemma}
\begin{proof} For every $z_0$ in the considered range, there exist its neighborhood in ${\mathbb C}$ with the closure inside this range. Let us consider $z$
in this neighborhood only. Then
$A(\sqrt{z},p)$ and its derivative with respect to $z$ and $p$  are continuous bounded functions of $z$ and $p$, so the holomorphicity of the matrix
element $(\psi, Z(z) \varphi)_{H^{1,2}}$ for every $\psi, \varphi \in H^{1,2}$ is immediately seen. So $Z$ is weakly holomorphic and then also strongly and in the norm
(Theorem 3.10.1 of \cite{Hille-Phillips} and slightly adapted Theorem VI.4 of \cite{RSIV}). The same argument apply in $L^2({\mathbb R})$.
\end{proof}

For $\delta\in {\mathbb R}$ let us denote a space
\begin{equation*}
W_\delta =\{ \varphi \in {\mathcal S}'({\mathbb R}) \, | \, {\mathcal F} \varphi \in L^2_\delta({\mathbb R}) \} = H^{\delta,2}({\mathbb R})
\end{equation*}
with the norm $\|\varphi\|_{W_\delta} = \|{\mathcal F}\varphi\|_{L^2_\delta}$.  Let us denote an operator $\hat{p}=-i\partial/\partial x$ acting
on a suitable domain in the space $L^2_\delta({\mathbb R})$.
\begin{lemma}
\label{p-z-1}
Let $K$ be a compact subset of the open upper complex half-plane $\{z\in{\mathbb C} \, | \, \Im{z} >0\}$ and $\delta\in {\mathbb R}$.
Then the operator $(\hat{p}-z)^{-1}$ is uniformly bounded in $L^2_\delta({\mathbb R})$ for $z\in K$.
\end{lemma}
\begin{proof} The operator $\hat{p}$ is self-adjoint in $L^2$ and therefore $(\hat{p}-z)^{-1}$ is uniformly bounded there by $1/y_0$ where
$y_0=\inf\{\Im z\, | \, z\in K\}>0$.
Let us calculate
\begin{eqnarray*}
[(\hat{p}-z)^{-1},w(x)^\delta]=(\hat{p}-z)^{-1}[w(x)^\delta,(\hat{p}-z)](\hat{p}-z)^{-1} =
\\
(\hat{p}-z)^{-1} \left(i\delta x w(x)^{\delta-2}\right) (\hat{p}-z)^{-1}
\end{eqnarray*}
which is uniformly bounded in $L^2$ for $\delta\leq 1$ (applying $[w^\delta,(\hat{p}-z)]$ first in ${\mathcal S}$ and then extending to $L^2$ by the density).

For $\psi\in L^2_\delta$, $0\leq \delta \leq 1$ then
\begin{eqnarray*}
\| (\hat{p}-z)^{-1}\psi \|_{L^2_\delta} = \|w^\delta (\hat{p}-z)^{-1}\psi \|_{L^2} =
\\
\| (\hat{p}-z)^{-1}w^\delta \psi + [w^\delta,(\hat{p}-z)^{-1}]\psi \|_{L^2} \leq
\\
\left(y_0^{-1} + \|[w^\delta,(\hat{p}-z)^{-1}]\|_{L^2\to L^2}\right) \|\psi\|_{L^2_\delta}
\end{eqnarray*}
remembering that $\|\psi\|_{L^2} \leq \|\psi\|_{L^2_\delta}$. So $(\hat{p}-z)^{-1}$ is uniformly bounded in $L^2_\delta$ for $0\leq \delta \leq 1$.

Assume now that $(\hat{p}-z)^{-1}$ is uniformly bounded in $L^2_\delta$ with respect to $z\in K$ for some $\delta\geq 0$ and $\psi\in L^2_{\delta + 1}$. Then
\begin{eqnarray*}
\|(\hat{p}-z)^{-1} \psi\|_{L^2_{\delta+1}} =
\\
\|(\hat{p}-z)^{-1}w^{\delta+1} \psi - i (\delta+1)(\hat{p}-z)^{-1}xw^{\delta-1}(\hat{p}-z)^{-1} \psi\|_{L^2} \leq
\\
\|(\hat{p}-z)^{-1}\|_{L^2 \to L^2}\left( \|\psi\|_{L^2_{\delta+1}}  + |\delta+1| \|(\hat{p}-z)^{-1}\|_{L^2_\delta \to L^2_\delta} \|\psi\|_{L^2_\delta} \right) \leq
\\
\|(\hat{p}-z)^{-1}\|_{L^2 \to L^2} \left(1  +|\delta+1| \|(\hat{p}-z)^{-1}\|_{L^2_\delta \to L^2_\delta}\right) \|\psi\|_{L^2_{\delta+1}}
\end{eqnarray*}
and $(\hat{p}-z)^{-1}$ is uniformly bounded in $L^2_{\delta+1}$ with respect to $z\in K$. By induction, it is now uniformly bounded in every $L^2_{\delta}$, $\delta\geq 0$.

Now we use the distribution like definition of $(\hat{p}-z)^{-1}$ in $L^2_{-\delta}$ and its uniform boundedness in $L^2_{\delta}$ extends to every $\delta \in {\mathbb R}$
by the duality. This completes the proof of the lemma.
\end{proof}
\begin{lemma}
\label{plambda2}
Let $\delta>\frac{3}{2}$. Then there exists a finite $c$ such that for every $\lambda\in {\mathbb C}$
and $\varphi\in {\mathcal S}({\mathbb R})$,
$$
\|\varphi\|_{W_{-\delta}} \leq c \|(p-\lambda)^2\varphi\|_{W_\delta}
$$
where
$p$ is a multiplication by the variable in ${\mathcal S}({\mathbb R})$.
\end{lemma}
\begin{proof} The statement is equivalent to the existence of $c$ such that
\begin{equation}
\label{phi_x_psi_ineq}
\|\varphi\|_{L^2_{-\delta}} \leq c \|(\frac{d}{dx}-\lambda)^2\varphi\|_{L^2_\delta}
\end{equation}
for every $\lambda\in {\mathbb C}$ and $\varphi\in {\mathcal S}({\mathbb R})$
which we prove following the proof of Lemma 3 in par. XIII.8 in \cite{RSIV}. Let us assume that $\lambda\in {\mathbb C}$ and $\varphi\in {\mathcal S}({\mathbb R})$ are given and put
$$
\psi = \left(\frac{d}{dx}-\lambda\right)^2\varphi \quad .
$$
Solving this equation we can express $\varphi$ through $\psi$ knowing that both are in ${\mathcal S}({\mathbb R})$.

Assume first the $\Re \lambda \leq 0$. Then
$$
\varphi(x)= \int_{-\infty}^x (x-y)\, e^{\lambda (x-y)}\, \psi(y) \, dy
$$
and
\begin{equation}
\label{phi_psi_estimate}
|\varphi(x)| \leq |x| \|\psi\|_{L^1} + \| y\psi\|_{L^1} \leq c (1+|x|) \|\psi\|_{L^2_\delta}
\end{equation}
with a constant $c$ dependent on $\delta >3/2$ but independent of $\lambda$ and $\varphi$. Here the relation
$$
\|y\psi\|_{L^1} = (w^{1-\delta},|y|w^{\delta-1}|\psi|) \leq \|w^{1-\delta}\|_{L_2} \|\psi\|_{L^2_\delta}
$$
and a similar one for $\|\psi\|_{L^1}$ were used.

For $\Re \lambda >0 $,
$$
\varphi(x) = \int^{+\infty}_x (y-x)\, e^{-\lambda (y-x)}\, \psi(y) \, dy
$$
and (\ref{phi_psi_estimate}) follows again directly giving (\ref{phi_x_psi_ineq}).
\end{proof}

The following lemma is an analog of Lemma 7 of par. XIII.8 in \cite{RSIV}.
\begin{lemma}
\label{R0lim}
Let $\delta > \frac{3}{2}$. Then there exists
\begin{equation}
\label{R00lim}
\lim_{y\to 0^+} (p-x-i y)^{-1}
\end{equation}
in the sense of ${\mathcal L}(W_\delta,W_{-\delta})$ uniformly with respect to $x\in {\mathbb R}$,
and
\begin{equation}
\label{R0inftylim} \lim_{y\to \pm\infty} (p-x-i y)^{-1} = 0
\end{equation}
in the sense of ${\mathcal L}(W_\delta,W_{-\delta})$ uniformly with respect to $x\in {\mathbb R}$.
\end{lemma}
\begin{proof} Let $0<y,y'<1$. Then
\begin{eqnarray*}
\| (p-x-i y')^{-1} -  (p-x-i y)^{-1} \|_{W_\delta\to W_{-\delta}} =
\\
\left\| \int_y^{y'} \frac{d}{dt}  (p-x-i t)^{-1} dt \right\|_{W_\delta\to W_{-\delta}}  \leq
\\
\left| \int_y^{y'} \|i  (p-x-i t)^{-2}  \|_{W_\delta\to W_{-\delta}} \, dt \right| \leq c |y' - y|
\end{eqnarray*}
according to Lemma \ref{plambda2} because the range of $(p-x-it)^2$ contains ${\mathcal S}({\mathbb R})$ for $t>0$. Then the limit (\ref{R00lim}) uniformly exists as ${\mathcal L}(W_\delta,W_{-\delta})$ is a Banach space.

For (\ref{R0inftylim}), the estimate
$$
|(p-x-iy)^{-1}|\leq |y|^{-1}
$$
is sufficient showing (\ref{R0inftylim})  in ${\mathcal L}(L^2)$, so it holds after Fourier transform in ${\mathcal L}(L^2)$ and then also in
${\mathcal L}(L^2_\delta, L^2_{-\delta})$  and by inverse Fourier transform in ${\mathcal L}(W_\delta,W_{-\delta})$.
\end{proof}
\begin{lemma}
\label{pzderivative}
Let $\delta>\frac{3}{2}$. Then the operator $(p-z)^{-1}$  is holomorphic in $z\in{\mathbb C}$, $\Im z >0$ with respect to ${\mathcal L}(W_\delta,W_{-\delta})$ norm,
its derivative being $(p-z)^{-2}$.
\end{lemma}
\begin{proof} By Lemma \ref{p-z-1}, $(p-z)^{-1} \in {\mathcal L}(W_\delta,W_{\delta}) \subset {\mathcal L}(W_\delta,W_{-\delta})$.
Let $\Im z_1>0$, $K$ be its bounded neighborhood with the closure inside the upper half-plane, $z_2 \in K$. We have
\begin{eqnarray*}
\left\| (z_2-z_1)^{-1} ((p-z_2)^{-1} - (p-z_1)^{-1}) - (p-z_1)^{-2} \right\|_{W_\delta \to W_{-\delta}} =
\\
|z_2 - z_1| \left\|(p-z_2)^{-1}(p-z_1)^{-2} \right\|_{W_\delta \to W_{-\delta}} \leq
\\
|z_2 - z_1| \left\|(p-z_2)^{-1}\right\|_{W_{-\delta} \to W_{-\delta}}  \left\|(p-z_1)^{-2} \right\|_{W_\delta \to W_{-\delta}}   \quad. \end{eqnarray*}
With the help of Lemmas \ref{p-z-1} and \ref{plambda2} the limit $z_2 \to z_1$ equals zero and the existence of the derivative at the point $z_1$ is seen.
\end{proof}
\begin{lemma}
\label{analprolong}
Let $\delta>\frac{3}{2}$. Then the operator function $z\mapsto (p-z)^{-1}$ can be analytically continued
in ${\mathcal L}(W_\delta,W_{-\delta})$ from the upper complex half-plane into the lower one.
\end{lemma}
\begin{proof} Let $\Im z_1 >0$ and define an operator function
\begin{eqnarray*}
F(z_2) = (p-z_1)^{-1} + \int_{z_1}^{z_2} (p-z)^{-2} \, dz
\end{eqnarray*}
where we integrate along a smooth curve crossing the real axis. Let us assume only one crossing for simplicity although it is not necessary.
By Lemmas \ref{p-z-1}-\ref{pzderivative} and their counterparts in the lower half-plane, $F(z_2)$ exists, equal $(p-z_2)^{-1}$ for $\Im z_2>0$, is holomorphic
in upper as well as in lower half-plane. It is also continuous in ${\mathbb C}$  by
Lemma \ref{plambda2} and independent of the choice of the integration curve. Then $F$ is holomorphic in ${\mathbb C}$ by Painlev\'{e}
theorem (e.g. Theorem 5.5.2 of \cite{BSimon2A}) applied to every $<\psi,F(z)\varphi>_{W_\delta,W_{-\delta}}$ with $\varphi,\psi \in W_\delta$, the holomorphicity
in norm following from the weak holomorphicity (e.g. Theorem 3.10.1 of \cite{Hille-Phillips} and Theorem VI.4 of \cite{RSIV}).
\end{proof}
\pagebreak[2]

After the preparatory considerations above, let us return to the study of operators closer connected to the resolvent of our Hamiltonian.
\begin{lemma}
\label{R0ubound}
Let $-\frac{\alpha^2}{4} < \lambda_1<\lambda_2<0$,
$\delta>\frac{3}{2}$. Then the operator
$(I-\alpha R_0(k))^{-1} : L^2_\delta({\mathbb R}) \to L^2_{-\delta}({\mathbb R})$ is uniformly bounded for
$k^2\in [\lambda_1,\lambda_2] + i [0,+\infty)$.
Here the above operator at $\Im k^2=0$ is defined as a limit from $\Im k^2>0$ and $R_0$ is defined in (\ref{R0def}).
\end{lemma}
\begin{proof} We shall use the Fourier transform (\ref{A_def}) and the relation
\begin{equation}
\label{Adecomposed}
A(k,p)=\frac{\alpha^2}{4 p_0}\left(\frac{1}{p-p_0}-\frac{1}{p+p_0}\right) +1 + \frac{\alpha}{2\sqrt{p^2-k^2} + \alpha}
\end{equation}
where
$$
p_0=\sqrt{k^2+\frac{\alpha^2}{4}} \quad , \quad \Re p_0>0 \quad.
$$
The existence of uniform limit at $\Im k^2=0$ now follows from Lemma \ref{R0lim} realizing that some terms in the above formula are bounded in the considered range of $k$.
The uniform bound in $k$ is similarly seen (notice that $\Re p_0 \to \infty$, $\Im p_0 \to \infty$ for $\Im k^2\to \infty$ and $\Re k^2$  bounded in $[\lambda_1,\lambda_2]$).
\end{proof}
\begin{lemma}
\label{R0analprolong}
Let $-\frac{\alpha^2}{4} < \lambda_1<\lambda_2<0$,
$\varepsilon_0 >0$,
$\delta>\frac{3}{2}$. Then the operator
$(I-\alpha R_0(k))^{-1} : L^2_\delta({\mathbb R}) \to L^2_{-\delta}({\mathbb R})$ can be analytically continued in ${\mathcal L}(L^2_\delta, L^2_{-\delta})$
from the upper half-plane to the region of $k^2 \in (\lambda_1,\lambda_2) + i (-\varepsilon_0,+\infty)$
\end{lemma}
\begin{proof} By Lemma \ref{Aholomorph}, the operator is holomorphic in $k^2 \in (\lambda_1,\lambda_2) + i (0,\infty)$ with respect to ${\mathcal L}(L^2, L^2)$ norm and so also with respect to ${\mathcal L}(L^2_\delta, L^2_{-\delta})$ norm. The statement again follows from the formula (\ref{Adecomposed}), Lemma \ref{analprolong} and its analogue on prolongation from the lower to
the upper half-plane. For the last term of (\ref{Adecomposed}) it is again possible to calculate the derivative explicitly in $L^2({\mathbb R},dp)$,  past back to $L^2({\mathbb R},ds)$
by the inverse Fourier transform and use inclusion ${\mathcal L}(L^2_\delta, L^2_{-\delta}) \subset {\mathcal L}(L^2, L^2)$ (remember that $\Re \sqrt{p^2-k^2} >0$ for all considered
$p$ and $k$).
\end{proof}
\begin{lemma}
\label{Rmm-R0compact}
For $k$ satisfying inequalities (\ref{kboundsb}), and $\delta>3$ from Assumption 1, $R_{mm}(k)-R_0(k)$ is a compact operator from $L^2_{-\delta/2}({\mathbb R})$
into $L^2_{\delta/2}({\mathbb R})$, uniformly bounded with respect to $k$.
\end{lemma}
\begin{proof} $R_{mm}(k)-R_0(k)$ is an integral operator with the kernel
\begin{equation}
\label{Rmm-R0kernel}
g(s,t)=\frac{1}{2\pi}K_0(-ik|\Gamma(s)-\Gamma(t)|) - \frac{1}{2\pi}K_0(-i k |s-t|)
\end{equation}
and it is sufficient to prove that the operator with the kernel $g_1(s,t)=w(s)^{\delta/2} g(s,t) w(t)^{\delta/2}$ is Hilbert-Schmidt in $L^2$, i.e. that
\begin{equation}
\label{Rmm-R0_HS}
\int_{\mathbb R} \int_{\mathbb R} w(s)^{\delta} |g(s,t)|^2 w(t)^{\delta}\,ds \,dt < \infty \quad,
\end{equation}
and that (\ref{Rmm-R0_HS}) is uniformly bounded.
It is sufficient to integrate only over $t<s$ due to the symmetry. Let $R>0$ be a number from Assumption 1. We consider separately the following ranges of variables $t<s$:
\begin{eqnarray*}
M_1:& -\infty<t<s<-R \quad,\\
M_2:& -\infty<t<-R<s<R \quad,\\
M_3:&  -\infty<t<-R<R<s<\infty \quad,\\
M_4:& -R<t<s<R \quad,\\
M_5:& -R<t<R<s<\infty  \quad,\\
M_6:& R<t<s< \infty  \quad.
\end{eqnarray*}

In the following, $c$ is a finite constant (independent of $k,s,t$ but depending on $\rho, \delta, \alpha, \lambda_1,\lambda_2,\varepsilon_0$)
which might have different values in various formulas (but one maximal finite value can be chosen). Estimates on Macdonald functions from Section
\ref{resolventsection} are used. In $M_1$ we estimate  (remember relations $K_0'=-K_1$, (\ref{kboundsb}), (\ref{K0derivative}) and Assumptions~1 and2)
\begin{eqnarray*}
|g(s,t)| \leq \frac{1}{2\pi} |k| \left||\Gamma(s)-\Gamma(t)|-|s-t|\right| \sup_{|\Gamma(s)-\Gamma(t)|\leq\xi\leq |s-t|}|K_1(-ik\xi)| \leq
\\
c |k| \left||\Gamma(s)-\Gamma(t)|-|s-t|\right| \frac{e^{-\frac{1}{2}k_{20}|\Gamma(s)-\Gamma(t)|}}{|k| |\Gamma(s)-\Gamma(t)|}  \leq
\\
c \left| | v_-(s-t) +\nu_-(s)-\nu_-(t)| -|s-t| \right| \frac{e^{-\frac{1}{2} k_{20} \rho |s-t|}}{|s-t|}  =
\\
c \left| | v_-(s-t) +\nu_-(s)-\nu_-(t)| -|v_-(s-t)| \right| \frac{e^{-\frac{1}{2} k_{20} \rho |s-t|}}{|s-t|}  \leq
\\
c \left| \nu_-(s)-\nu_-(t) \right| \frac{e^{-\frac{1}{2} k_{20} \rho |s-t|}}{|s-t|}  \leq c \varphi_-(s) e^{-\frac{1}{2} k_{20} \rho |s-t|} \quad.
\end{eqnarray*}
Now
\begin{eqnarray}
\nonumber
\int_{M_1} |g_1(s,t)|^2 \,ds\,dt   \leq
\\
\nonumber
c \int_{-\infty}^{-R} w(s)^{\delta} \varphi_-(s)^2 \int_{-\infty}^s e^{-k_{20}\rho(s-t)} w(t)^{\delta} \, dt \,ds  =
\\
\nonumber
c \int_{-\infty}^{-R} w(s)^{\delta} \varphi_-(s)^2 \int_0^{\infty}  e^{-k_{20}\rho \tau } w(|s|+ \tau )^{\delta} \, d\tau \, ds    \leq
\\
\label{M1int}
c \int_{-\infty}^{-R} w(s)^{\delta}(1+w(s)^\delta) \varphi_-(s)^2 \, ds
\end{eqnarray}
where inequality $w(|s|+\tau) \leq w(s)+w(\tau)$ and convexity of the function $x\mapsto x^\delta$ were used ($c$ is changed by a factor in these estimates).
The convergence of (\ref{M1int}) then follows from the assumption (\ref{phiLdelta}). The integral over $M_6$ is treated in the same way.

For the case of $M_2$ let us similarly estimate
\begin{eqnarray*}
|g(s,t)| \leq c \frac{ |s-t| - |\Gamma(s)-\Gamma(t)| }{|s-t|} e^{-\frac{1}{2} k_{20} \rho |s-t|} \leq c e^{-\frac{1}{2} k_{20} \rho (s-t)}
\end{eqnarray*}
and
\begin{eqnarray*}
\int_{M_2} |g_1(s,t)|^2 \, ds \, dt \leq \left(\int_{-R}^R e^{-k_{20}\rho s} w(s)^\delta \, ds\right)
\left(\int_{-\infty}^{-R} e^{k_{20}\rho t} w(t)^\delta \, dt\right) < \infty \quad.
\end{eqnarray*}
The remaining parts $M_3, M_4, M_5$ can be treated analogically.
\end{proof}
\noindent
{\bf Remark.} For the validity of the Lemma \ref{Rmm-R0compact}, assumption (\ref{phiLdelta}) with $\delta \geq 1$ is sufficient as is seen from the proof.
\begin{lemma}
For $\delta>3$ from Assumption 1 and $k$ satisfying (\ref{kbounds})
\begin{equation*}
(R_{mm}(k)-R_0(k)) (I-\alpha R_0(k))^{-1}
\end{equation*}
is a compact operator in $L^2_{\delta/2}$, uniformly bounded in norm with respect to $k$ in the considered range.
\end{lemma}
\begin{proof} As $\delta/2 >3/2$ the operator
$(I-\alpha R_0(k))^{-1}$ from  $L^2_{\delta/2}$ into $L^2_{-\delta/2}$ is uniformly bounded by Lemma \ref{R0ubound}.
The operator $(R_{mm}(k)-R_0(k))$ from $L^2_{-\delta/2}$ into $L^2_{\delta/2}$ is compact and uniformly bounded by Lemma \ref{Rmm-R0compact}. So their product is compact and uniformly bounded.
\end{proof}
\begin{lemma}
\label{Rmm-R0boundary}
Let
$\delta$ be the number from Assumption 1 ($\delta \geq 1$ is sufficient), $\lambda_1 < \lambda_2 < 0$, $k^2=\lambda + i\varepsilon$, $\Im k >0$. Then
$$
\lim_{\lambda\to\lambda_0,\varepsilon\to 0^+} (R_{mm}(k) - R_0(k))
$$
in the $L^2_{-\delta/2}({\mathbb R}) \to L^2_{\delta/2}({\mathbb R})$ operator norm exists uniformly for $\lambda_0 \in [\lambda_1,\lambda_2]$.
\end{lemma}
\begin{proof} Let $k^2=\lambda + i\varepsilon$, $k'^2=\lambda' + i\varepsilon'$. We want to estimate the difference
$$
(R_{mm}(k') - R_0(k')) - (R_{mm}(k) - R_0(k))
$$
in the $L^2_{-\delta/2} \to L^2_{\delta/2}$ operator norm for which it is sufficient to estimate
$$
w^{\delta/2} ((R_{mm}(k') - R_0(k')) - (R_{mm}(k) - R_0(k))) w^{\delta/2}
$$
in the Hilbert-Schmidt norm. We start from the estimate of the kernel denoting $[k,k']$ the interval from $k$ to $k'$ in the complex plain,
using Assumptions 1 and 2, estimates from the section \ref{resolventsection} and (\ref{kbounds}). We obtain
\begin{eqnarray*}
\left|K_0(-ik'|\Gamma(s)-\Gamma(t)|) - K_0(-ik'|s-t|) - \right.
\\
\left. K_0(-ik|\Gamma(s)-\Gamma(t)|) + K_0(-ik|s-t|)\right| \leq
\\
|k'-k| \sup_{\xi\in [k,k']} \left| (|\Gamma(s)-\Gamma(t)|-|s-t|) K_1(-i\xi |\Gamma(s)-\Gamma(t)|) + \right.
\\
\left. |s-t|(K_1(-i\xi |\Gamma(s)-\Gamma(t)|) - K_1(-i\xi |s-t|)) \right| \leq
\\
|k'-k| ||\Gamma(s)-\Gamma(t)|-|s-t|| \left( \sup_{\xi\in [k,k']} \left|K_1(-i\xi |\Gamma(s)-\Gamma(t)|)\right| + \right.
\\
\left. |s-t| \sup_{\xi\in [k,k']} |\xi| \sup_{\eta\in [ |\Gamma(s)-\Gamma(t)|,|s-t|]} \left|K_1'(-i\xi\eta)\right| \right) \leq
\\
c |k'-k|\left( \frac{|s-t|-|\Gamma(s)-\Gamma(t)|}{|s-t|} \right) e^{-\frac{1}{2}k_{20}\rho |s-t|} \quad.
\end{eqnarray*}
Now we have to estimate
\begin{equation}
\label{HS_estim1}
d=\int_{{\mathbb R}^2} w(s)^\delta \left( \frac{|s-t|-|\Gamma(s)-\Gamma(t)|}{|s-t|} \right)^2 e^{-k_{20}\rho |s-t|} w(t)^\delta \, ds\, dt \quad.
\end{equation}
As in the proof of Lemma \ref{Rmm-R0compact}, we integrate separately in the regions $M_1,\dots,M_6$ some cases being completely same and some very slightly different. In $M_1$ and $M_6$, the $\Gamma$-dependent term in the brackets is estimated with the help of $\varphi_\pm$, in $M_2,\dots,M_5$, estimate by the constant $1$ is sufficient. We see that $d<\infty$, so $R_{mm}(k)-R_0(k)$ is uniformly Cauchy in $k$ and the required uniform limit exists.
\end{proof}
\pagebreak[2]
\begin{lemma}
\label{Rmm-R0holomrphic}
Under the assumptions of Lemma \ref{Rmm-R0boundary}, operator $R_{mm}(k)-R_0(k)$ from $L^2_{-\delta/2}({\mathbb R})$ into $L^2_{\delta/2}({\mathbb R})$
is holomorphic in $k^2\in (\lambda_1,\lambda_2) +i(-\varepsilon_0,\varepsilon_0)$.
\end{lemma}
\begin{proof} It is equivalent to show that the considered operator function is holomorphic in $k$. As in the proofs of previous lemmas, let us consider
the kernel $g(k,s,t)$ (\ref{Rmm-R0kernel})
of the operator $(R_{mm}(k) - R_0(k))$,  and expand
$$
g(k',s,t)-g(k,s,t)=(k'-k)\frac{\partial g(k,s,t)}{\partial k} + r(k',k,s,t) \quad.
$$
Let us write explicitly
\begin{eqnarray*}
2\pi \frac{\partial g(k,s,t)}{\partial k}=
-i|\Gamma(s)-\Gamma(t)| K_0'(-ik |\Gamma(s)-\Gamma(t)|) + i|s-t| K_0'(-ik |s-t|) \quad,
\\
2\pi \frac{\partial^2 g(k,s,t)}{\partial k^2}=
-|\Gamma(s)-\Gamma(t)|^2 K_0''(-ik |\Gamma(s)-\Gamma(t)|) + |s-t|^2 K_0''(-ik |s-t|) \quad.
\end{eqnarray*}
With the estimate (\ref{K0derivative}) and (\ref{kbounds}), we obtain
\begin{eqnarray*}
2\pi \left|\frac{\partial g(k,s,t)}{\partial k}\right| \leq
||\Gamma(s)-\Gamma(t)| - |s-t|| |K_0'(-ik|\Gamma(s)-\Gamma(t)|)| +
\\
|s-t| |K_0'(-ik|\Gamma(s)-\Gamma(t)|) - K_0'(-ik|s-t|)| \leq
\\
c \frac{|s-t| - |\Gamma(s)-\Gamma(t)|}{|s - t|} e^{-\frac{1}{2}k_{20}\rho |s-t|}
\end{eqnarray*}
and the convergence of (\ref{HS_estim1}) show that kernel $\partial g/ \partial k$ defines a bounded operator $L^2_{-\delta/2}\to L^2_{\delta/2}$.

Let us estimate
$$
|r(k',k,s,t)| \leq |k'-k|^2 \sup_{\xi\in [k,k']} \left|\frac{\partial^2 g(\xi,s,t)}{\partial k^2}\right| \quad.
$$
Similarly as above
\begin{eqnarray*}
2\pi\left|\frac{\partial^2 g(\xi,s,t)}{\partial k^2}\right| \leq \left( |s-t|^2 - |\Gamma(s)-\Gamma(t)|^2\right) \left| K_0''(-i\xi |\Gamma(s)-\Gamma(t)|) \right| +
\\
|s-t|^2 \left| K_0''(-i\xi |\Gamma(s)-\Gamma(t)|) - K_0''(-i\xi |s-t|) \right| \leq
\\
c \frac{|s-t| - |\Gamma(s)-\Gamma(t)|}{|s - t|} e^{-\frac{1}{2}k_{20}\rho |s-t|}
\end{eqnarray*}
and convergency of (\ref{HS_estim1}) again show that the kernel $(k'-k)^{-2}r(k',k,s,t)$ defines operator $L^2_{-\delta/2}\to L^2_{\delta/2}$
uniformly bounded in $k',k$. So the existence of
$\frac{\partial}{\partial k} (R_{mm}(k)-R_0(k))$ follows.
\end{proof}
\begin{lemma}
\label{Rmm-R0infty}
Let $\delta\geq 1$ be the number from Assumption 1, $k^2=\lambda +i\varepsilon$, $\lambda < 0$, $\varepsilon >0$, $\Im k >0$.
Then
$$
\lim_{\varepsilon \to +\infty} \| R_{mm}(k)-R_0(k) \|_{L^2_{-\delta/2}\to L^2_{\delta/2}} = 0 \quad.
$$
\end{lemma}
\begin{proof} As above, it is sufficient to show that $w^{\delta/2} (R_{mm}(k)-R_0(k)) w^{\delta/2}$ tends to zero in the Hilbert-Schmidt norm.
We use the estimates of its kernel from the proof of Lemma \ref{Rmm-R0compact} where the overall constant $c$ is independent of $k$ and
$k_{20}$ may be replaced by $\sqrt{\varepsilon/2}$. The estimate (\ref{K0derivative}) can be used as direct calculations show that
$k_1/k_2 \to 1$ as $\varepsilon \to +\infty$ with $\lambda$ fixed.
Then we use dominated convergence to finish the proof.
\end{proof}
\begin{lemma}
\label{small_inverse_lemma}
Let $\delta>3$ be the number from Assumption 1, $\alpha>0$. Then there exists a discrete subset ${\mathcal E}$ of {\emph open} interval $(-\frac{\alpha^2}{4},0)$
with the property that for any closed interval $[\lambda_1,\lambda_2] \subset (-\frac{\alpha^2}{4},0) \setminus {\mathcal E}$  there exists $\varepsilon_0>0$ such that
\begin{equation}
\label{small_inverse}
\left(I-\alpha (R_{mm}(k)-R_0(k))(I-\alpha R_0(k))^{-1} \right)^{-1}
\end{equation}
exists and is uniformly bounded in  $L^2_{\delta/2}({\mathbb R}) \to L^2_{\delta/2}({\mathbb R})$ operator norm for
$k^2=\lambda +i\varepsilon$, $\lambda\in [\lambda_1,\lambda_2]$, $0<\varepsilon<\varepsilon_0$, $\Im k >0$.
\end{lemma}
\begin{proof}  By Lemmas \ref{Rmm-R0compact} and \ref{Rmm-R0holomrphic}, $R_{mm}(k)-R_0(k)$ is uniformly bounded, compact and holomorphic as
$L^2_{-\delta/2} \to L^2_{\delta/2}$ operator for $k^2\in [\lambda_1',\lambda_2'] +i [-\varepsilon_1,\varepsilon_1]$ with arbitrary fixed
$-\frac{\alpha^2}{4} < \lambda_1' < \lambda_2'<0$, $\varepsilon_1>0$.
By Lemma \ref{R0analprolong}, $(I-\alpha R_0(k))^{-1}$ is uniformly bounded and holomorphic as $L^2_{\delta/2} \to L^2_{-\delta/2}$ operator there.
Now
$$
\alpha (R_{mm}(k)-R_0(k)) (I-\alpha R_0(k))^{-1}
$$
is uniformly bounded compact holomorphic operator $L^2_{\delta/2}\to L^2_{\delta/2}$.

With the help of Lemma \ref{Rmm-R0infty}, we also know that
$$
\lim_{\varepsilon\to +\infty} \|(R_{mm}(k)-R_0(k))(I-\alpha R_0(k))^{-1}\|_{L^2_{\delta/2} \to L^2_{\delta/2}} =0 $$
and then (\ref{small_inverse}) exists for $\varepsilon$ large enough.
By the analytic Fredholm theorem (e.g. Theorem VI.14 in \cite{RSIV}), (\ref{small_inverse}) now exists and is holomorphic for
$k^2\in ((\lambda'_1,\lambda_2') +i (-\varepsilon_1, +\varepsilon_1)) \setminus {\mathcal E_1}$  where  ${\mathcal E_1}$ is a discrete subset of
$(\lambda'_1,\lambda_2') +i (-\varepsilon_1, +\varepsilon_1)$. Let ${\mathcal E} \cap (\lambda_1',\lambda_2')$ be now the intersection of ${\mathcal E_1}$ with the real axis.
This defines ${\mathcal E}$ discrete in $(-\frac{\alpha^2}{4},0)$ as $\lambda_1', \lambda_2'$ were arbitrary.
For $[\lambda_1,\lambda_2] \subset (-\frac{\alpha^2}{4},0) \setminus {\mathcal E}$ it is now possible to choose $\varepsilon_0>0$ such that
$([\lambda_1,\lambda_2] +i [0,\varepsilon_0]) \cap {\mathcal E_1} = \emptyset$ and the lemma is proved.
\end{proof}
\pagebreak[2]
\begin{lemma}
\label{Rphidelta}
Let $\varphi \in C_0^\infty({\mathbb R}^2)$, $k$ satisfies (\ref{kbounds}), $\delta\geq 0$ and $R_{m\, dx}(k)$ be defined in (\ref{Rmx}).
Then $R_{m\, dx}(k)\varphi \in L^2_\delta({\mathbb R})$ with the norm uniformly bounded with respect to $k$.
\end{lemma}
\begin{proof} Using (\ref{Rmx}), (\ref{K0bound}), (\ref{kbounds}) and denoting $c$ a suitable constant,
\begin{eqnarray*}
|(R_{m\, dx}(k)\varphi)(s)| \leq c \|\varphi\|_\infty \int_{{\rm supp}\, \varphi}(1+ |\log(|k| |\Gamma(s)-y|)|) e^{-k_{20}|\Gamma(s)-y|} \, d^2y \leq
\\
c \|\varphi\|_\infty e^{k_{20}d} e^{-k_{20} |\Gamma(s)-\Gamma(0)|}\int_{{\rm supp}\, \varphi}  ( 1 + |\log |k|| +|\log |\Gamma(s)-y||)\, d^2y
\end{eqnarray*}
where
$$
d=\sup_{y\in {\rm supp}\, \varphi} |y-\Gamma(0)| < \infty \quad.
$$
As $0<k_{20}\leq |k| \leq k_{21}+b < \infty$, term $|\log |k||$ is bounded by a constant in the above expression. Further,
$$
|\Gamma(s)-y| \leq |\Gamma(s)-\Gamma(0)| + |\Gamma(0)-y| \leq |s| + d \quad.
$$
If $|\Gamma(s)-\Gamma(0)|\geq d+1$ then  $|\Gamma(s)-y|\geq 1$ and
\begin{eqnarray*}
0\leq \log|\Gamma(s)-y| \leq \log(|s|+d) \quad,
\\
\int_{{\rm supp}\, \varphi} |\log|\Gamma(s)-y|| \, d^2y \leq c_1 \log(|s|+d)
\end{eqnarray*}
where $c_1$ is a suitable constant (dependent on $\varphi$).

If $|\Gamma(s)-\Gamma(0)| < d+1$, then denoting $B(a,r)$ the disk in ${\mathbb R}^2$ of the center $a$ and radius $r$
\begin{eqnarray*}
\int_{{\rm supp}\, \varphi} |\log|\Gamma(s)-y||\, d^2y \leq \int_{B(\Gamma(s), 2 d + 1)} |\log|\Gamma(s)-y||\, d^2y =
\\
\int_{B(0, 2d+1)} |\log|y||\, d^2y = c_2<\infty \quad.
\end{eqnarray*}
Finally, we can estimate for every $s$
$$
|(R_{m\, dx}(k)\varphi)(s)| \leq c_3 e^{-k_{20}\rho |s|} (1 + |\log(|s|+d)|)
$$
with a suitable finite $c_3$ (dependent on $\varphi$) and $R_{m\, dx}(k)\varphi \in L^2_\delta$.
\end{proof}

Now we give the first main result.
\begin{theorem}
\label{spectralac}
Under the Assumptions 1-4, there exists at most discrete subset ${\mathcal E}$ of open interval $(-\frac{\alpha^2}{4},0)$ such that
\begin{eqnarray*}
\left(-\frac{\alpha^2}{4},0\right) \cap \sigma_{pp}(H) = {\mathcal E}    \quad,
\\
\left(-\frac{\alpha^2}{4},0\right) \subset \sigma_{ac}(H) \quad,
\\
\left(-\frac{\alpha^2}{4},0\right) \cap \sigma_{sc}(H) = \emptyset \quad.
\end{eqnarray*}
\end{theorem}
\begin{proof} By Proposition \ref{essspectr}, $(-\frac{\alpha^2}{4},0) \subset \sigma_{ess}(H)$.
Let ${\mathcal E}$ be the set from Lemma \ref{small_inverse_lemma} and $[\lambda_1,\lambda_2] \subset (-\frac{\alpha^2}{4},0) \setminus {\mathcal E}$.
We show that $(\lambda_1,\lambda_2) \subset \sigma_{ac}(H)$ and $(\lambda_1,\lambda_2) \cap (\sigma_{pp}(H) \cup \sigma_{sc}(H)) =\emptyset$.
By Theorem XIII.20 of \cite{RSIV}, it is sufficient to show that
\begin{equation}
\label{integral_bound}
\sup_{0<\varepsilon<\varepsilon_0} \int_{\lambda_1}^{\lambda_2} |\Im (\varphi, (H-\lambda-i\varepsilon)^{-1} \varphi)|^p \, d\lambda < \infty
\end{equation}
for some $p>1$, $\varepsilon_0>0$ and every $\varphi \in C_0^\infty({\mathbb R}^2)$.
As above, we shall write $k^2=\lambda+i\varepsilon$, $\Im k >0$. We use formula (\ref{full_resolvent}),
\begin{eqnarray*}
(H-\lambda-i\varepsilon)^{-1} = R_{dx\, dx}(k) + \alpha R_{dx\, m}(k) (I-\alpha R_{mm}(k))^{-1} R_{m\, dx}(k) =
\\
R_{dx\, dx}(k) + \alpha R_{dx\, m}(k) (I-\alpha R_0(k))^{-1} \cdot
\\
\cdot
\left( I - \alpha (R_{mm}(k) - R_0(k)) (I-\alpha R_0(k))^{-1}\right)^{-1} R_{m\, dx}(k)   \quad.
\end{eqnarray*}
As $(\lambda_1,\lambda_2)$ is contained in the resolvent set of the free Laplacian,
$$
|(\varphi,R_{dx\, dx}(k)\varphi)| \leq |\lambda_2|^{-1} \|\varphi\|_{L^2({\mathbb R}^2)}^2 \quad.
$$
Let $\delta>3$ be the number from Assumption 1. Then by Lemma \ref{Rphidelta}, $\|R_{dx\, m}(k)\varphi\|_{L^2_{\delta/2}({\mathbb R})}$
is uniformly bounded with respect to $k$. By Lemmas \ref{R0ubound} and \ref{small_inverse_lemma}
\begin{eqnarray*}
R_2 \varphi :=
\\
(I-\alpha R_0(k))^{-1} \left( I - \alpha (R_{mm}(k) - R_0(k)) (I-\alpha R_0(k))^{-1}\right)^{-1} R_{m\, dx}(k) \varphi
\end{eqnarray*}
is uniformly bounded in $L^2_{-\delta/2}({\mathbb R})$.
Now
\begin{eqnarray*}
2\pi |(\varphi, R_{dx\, m}(k) R_2\varphi)|=
\\
\left| \int_{{\mathbb R}^2} \overline{\varphi(x)} \left(\int_{\mathbb R} K_0(-ik|x-\Gamma(s)|)(R_2\varphi)(s)\, ds\right)\, d^2x \right|  \leq
\\
\left(\int_{{\mathbb R}^2} \int_{\mathbb R}|\varphi(x)|^2 |K_0(-ik|x-\Gamma(s)|)|^2 w(s)^{\delta} \, d^2 x\, ds\right)^{1/2}
\cdot
\\
\cdot
\left(\int_{{\rm supp}\, \varphi} d^2x \right)^{1/2}  \|R_2\varphi\|_{L^2_{-\delta/2}({\mathbb R})}
\end{eqnarray*}
and practically the same estimates as in the proof of Lemma \ref{Rphidelta} show that this is uniformly bounded in $k$.
Now
$$
|(\varphi,(H-\lambda -i \varepsilon)^{-1}\varphi)|
$$
is uniformly bounded with respect to $\lambda\in (\lambda_1,\lambda_2)$ and $\varepsilon \in (0,\varepsilon_0)$ and (\ref{integral_bound}) is verified for any $p>1$.

So we have
\begin{eqnarray*}
\left(-\frac{\alpha^2}{4},0\right) \cap \sigma_{pp}(H) \subset {\mathcal E}    \quad,\quad
\left(-\frac{\alpha^2}{4},0\right) \setminus {\mathcal E}\subset \sigma_{ac}(H) \quad,
\\
\left( \left(-\frac{\alpha^2}{4},0\right) \setminus {\mathcal E} \right) \cap \sigma_{sc}(H) = \emptyset \quad.
\end{eqnarray*}
However, isolated points of ${\mathcal E}$ which are not
eigenvalues of $H$ can be skipped from ${\mathcal E}$ without change of the statements and the theorem is proved.
\end{proof}
\noindent
{\bf Remark}. Theorem \ref{spectralac} does not exclude that the points $-\frac{\alpha^2}{4}$ or $0$ are accumulation points of embedded eigenvalues. This would
desire a further  study as well as the existence of the embedded eigenvalues in general.
\section{Wave operators}
\label{wosection}
We turn to the study of scattering for particles moving along the curve $\Gamma$ (\ref{nudef}) with the energies in $(-\frac{\alpha^2}{4},0)$ and the asymptotic states corresponding
to the movement along the straight lines $\Gamma^{(\pm)}$,
$$
\Gamma^{(\pm)}(s) = a_\pm + v_\pm s, \quad s \in {\mathbb R}.
$$
We have two directions of the asymptotic movement $v_\pm$ in each of which the particles can be incoming or outgoing and we need four wave operators. We include some projectors into
their definitions to distinguish these cases (cf. \cite{RSIV}, vol. 3, par. XI.3). Let us denote as $H^{(\pm)}$ the Hamiltonians corresponding to the curves $\Gamma^{(\pm)}$
and $J^{(\pm)}$ their spectral projectors to the energy interval $(-\frac{\alpha^2}{4},0)$.
Remember that $H^{(\pm)}$ has the absolutely continuous spectrum only.
Further, denote $P^{(\pm)}_>$ the spectral projectors to the interval $(0,+\infty)$ of the
momentum in the direction $v_\pm$, i.e., of the self-adjoint operator $-i v_\pm \cdot \nabla$. Similarly, $P^{(\pm)}_<$ denotes the spectral projector to the interval $(-\infty, 0)$
of the momentum in the direction $v_\pm$, i.e. the movement in the direction $-v_\pm$. We are interested in the wave operators
\begin{eqnarray}
\label{wo_first}
\Omega^{(+)}_{in} = s-\lim_{t\to -\infty} e^{i H t} P^{(+)}_< e^{-i H^{(+)} t} J^{(+)} \quad ,
\\
\label{wo}
\Omega^{(+)}_{out} = s-\lim_{t\to +\infty} e^{i H t} P^{(+)}_> e^{-i H^{(+)} t} J^{(+)} \quad ,
\\
\Omega^{(-)}_{in} = s-\lim_{t\to -\infty} e^{i H t} P^{(-)}_> e^{-i H^{(-)} t} J^{(-)} \quad ,
\\
\label{wo_last}
\Omega^{(-)}_{out} = s-\lim_{t\to +\infty} e^{i H t} P^{(-)}_< e^{-i H^{(-)} t} J^{(-)} \quad .
\end{eqnarray}
Their interpretation is clear, e.g. $\Omega^{(+)}_{in}$ corresponds to particles incoming along the curve in the direction opposite to $v_+$. $S$-matrix is formed from these wave operators
in the usual way. The other four possible wave operators with interchanged $P^{(\pm)}_>$ and $P^{(\pm)}_<$ are not of the physical interest as they would correspond to the particles
outgoing in the past or incoming in the future.

We prove the existence of $\Omega^{(+)}_{out}$. The existence of the other wave operators can be then proved in the same way or by the symmetry arguments. If $\Omega^{(+)}_{out}$ exists for every curve
$\Gamma$ satisfying our assumptions the existence of $\Omega^{(-)}_{out}$ is seen using the reparametrization $\tilde{\Gamma}(s)=\Gamma(-s)$ and the existence of $\Omega^{(+)}_{in}$ and $\Omega^{(-)}_{in}$
then follows by the time-reversal invariance.

For the simplicity of notation we choose the coordinate system such that the first axis is oriented along $v_+$, i.e.
\begin{equation}
\label{specialv+}
v_+=(1,0) \quad.
\end{equation}
Now
\begin{equation}
\label{H+def}
H^{(+)} = \overline{ -\partial_{x_1}^2 \otimes I_2 + I_1 \otimes H_{2,\alpha} }
\end{equation}
where $I_{1,2}$ are identity operators in $L^2({\mathbb R})$, $\partial_{x_1}^2$ means the free Laplacian in one dimension and $H_{2,\alpha}$ is the one dimensional Schr\"{o}dinger operator
with the $-\alpha \delta(x_2)$ interaction. $H_{2,\alpha}$ has one eigenvalue $-\frac{\alpha^2}{4}$
with the eigenfunction
\begin{equation}
\varphi_0(x_2) = \sqrt{\frac{\alpha}{2}} e^{-\frac{\alpha}{2} |x_2|}
\end{equation}
and absolutely continuous spectrum $[0,+\infty)$, e.g. \cite{Albeverio}, Chapter I.3. The spectral measure of the operator $H^{(+)}$ with separated variables is the tensor convolution of the spectral measures
for $-\partial_{x_1}^2$ and $H_{2,\alpha}$ \cite{Fox}. So we can explicitly construct the projector in (\ref{wo}) (notice that the last three operators in the product commute there)
\begin{equation}
\label{spectproj}
P^{(+)}_>J^{(+)} = {\mathcal F_1}^{-1} \chi_{(0,\alpha/2)} {\mathcal F_1} \varphi_0 (\varphi_0,\cdot)_2
\end{equation}
where $(\cdot,\cdot)_2$ is the scalar product in $L^2({\mathbb R}, d x_2)$ used for the projection on $\varphi_0(x_2)$, ${\mathcal F}_1$ is the Fourier transform in the variable $x_1$
and $\chi_{(0,\alpha/2)}$ the characteristic function projecting on the range of the corresponding momentum variable $p_1$ in the interval $(0,\alpha/2)$.

We prove the existence of the wave operator $\Omega^{(+)}_{out}$ using generalized Kuroda-Birman theorem given in the Appendix B. We use the resolvent in the form (\ref{full_resolvent})
with $k=i\kappa$, $\kappa>0$ large enough for the validity of the following considerations. Then $z=k^2=-\kappa^2$ belongs to the resolvent sets of both $H$ and $H^{(+)}$ and $\|\alpha R_{mm}\| <1$ so that
\begin{equation}
\label{reduced_resolvent}
(I-\alpha R_{mm}(k))^{-1} = I + r
\end{equation}
where $r$ is a bounded operator in $L^2({\mathbb R})$. We show now that $r$ is an integral operator and derive some bounds on its kernel. We use letter $c$ to denote a constant which might have
different values in different formulas. Let $f\in L^2({\mathbb R})$, then
\begin{eqnarray*}
|(\alpha R_{mm}(k) f)(s)| \leq c \int_{\mathbb R} K_0(\kappa \rho |s-t|) |f(t)| \, dt
\\
\leq c \int_{\mathbb R} {\mathcal K}_1(s-t) |f(t)| \, dt \quad,
\\
{\mathcal K}_1(s)=c (1+ |\log \kappa\rho |s||) e^{-\kappa \rho |s|}
\end{eqnarray*}
according to (\ref{Rmm}), (\ref{rhodef}) and (\ref{K0bound}). By the Young inequality (e.g. Sect. IX.4 of \cite{RSIV}), $\|\alpha R_{mm}(k)\| \leq c/(\kappa \rho) < 1$ for $\kappa$ sufficiently large.
Now
$$
(rf)(s) = \int_{\mathbb R} {\mathcal R}(s,t) f(t) \, dt
$$
where
\begin{equation}
\label{kernelK}
|{\mathcal R}(s,t)| \leq {\mathcal K}(s-t) \quad,\quad {\mathcal K} = {\mathcal K}_1 + {\mathcal K}_1*{\mathcal K}_1 + {\mathcal K}_1*{\mathcal K}_1*{\mathcal K}_1 + \dots
\end{equation}
provided that the last series is convergent in $L^1({\mathbb R})$ which is the case for $\kappa$ large enough. Even more,
${\mathcal K}_1 \in L^1({\mathbb R}, (1+|s|) ds) =: {\mathcal H}_1$
and  $\| {\mathcal K}_1\|_{{\mathcal H}_1} < 1$ for $\kappa$ large enough. Now $\|{\mathcal K}_1*\dots * {\mathcal K}_1\|_{{\mathcal H}_1} \leq \|{\mathcal K}_1\|_{{\mathcal H}_1}^n$
for $n-1$ convolutions%
\footnote{We could use $L^1_{1/2}({\mathbb R})$ instead of ${\mathcal H}_1$ but then an unessential factor of the form $c^{n-1}$ would appear in the last formula.
The both spaces are of course equal as sets.}
and we see that ${\mathcal K}\in{\mathcal H}_1$.

Let now $K=P^{(+)}_> J^{(+)}$ (\ref{spectproj}) and consider the difference
\begin{eqnarray*}
T&=& K (H^{(+)}+\kappa^2)^{-1}-(H+\kappa^2)^{-1}K \quad,
\\
&=& \left( (H^{(+)}+\kappa^2)^{-1}-(H+\kappa^2)^{-1} \right) K \quad,
\\
T&=&C_1+C_2 \quad,
\\
C_1&=&\left( (H^{(+)}+\kappa^2)^{-1}-(H+\kappa^2)^{-1} \right) \theta(x_1) K \quad,
\\
C_2&=&\left( (H^{(+)}+\kappa^2)^{-1}-(H+\kappa^2)^{-1} \right) \theta(-x_1) K \quad.
\end{eqnarray*}

Let
\begin{eqnarray*}
M_1={\mathcal F_1^{-1}} C_0^\infty ((-\infty,0) \cup (0,\frac{\alpha}{2}) \cup (\frac{\alpha}{2},+\infty))
\quad,\quad
M_2={\mathcal M}(H_{2,\alpha})
\quad,
\\
M= \{M_1 \times M_2\}_{lin}
\end{eqnarray*}
where the definition of ${\mathcal M}$ is reminded in the Appendix B and $M$ is the set of all finite linear combinations of
vectors $\varphi_1 \otimes \varphi_2$, $\varphi_1\in M_1$, $\varphi_2 \in M_2$ (sometimes called as the algebraic tensor product).
As the closures $\overline{M_1} = L^2({\mathbb R})$, $\overline{M_2} = L^2({\mathbb R})$ the set $M$ is dense, $\overline{M}=L^2({\mathbb R}^2)$.
\begin{lemma}
The above defined set
$$
M \subset {\mathcal M}(H^{(+)}) \quad.
$$
\end{lemma}
\begin{proof} The Stone formula easily leads to
$$
d (\varphi_1, E^{H_1}_\lambda \varphi_1) = \left( |\hat{\varphi_1}(\sqrt{\lambda})|^2 + |\hat{\varphi_1}(-\sqrt{\lambda})|^2 \right) \frac{\theta(\lambda)}{2\sqrt{\lambda}}\, d\lambda
$$
where $H_1$ is the self-adjoint operator corresponding to $-\partial_{x_1}^2$ in $L^2({\mathbb R})$, $E^{H_1}_\lambda$ its spectral projector and $\hat{\varphi_1} = {\mathcal F}_1 \varphi_1$.
So $M_1 \subset {\mathcal M}(H_1)$. Now $\varphi_1 \otimes \varphi_2 \in {\mathcal M}(H^{(+)})$ for every $\varphi_1 \in M_1$, $\varphi_2 \in M_2$.
To see that assume
$$
d(\varphi_1, E^{H_1}_\lambda \varphi_1) = f_1(\lambda) d\lambda \quad,\quad d(\varphi_2, E^{H_{2,\alpha}}_\lambda \varphi_2) = f_2(\lambda) d\lambda
$$
with $f_{1,2} \in L^1({\mathbb R})\cap L^\infty({\mathbb R})$. Then \cite{Fox}
\begin{eqnarray*}
(\varphi_1 \otimes \varphi_2, E^{H^{(+)}}_\lambda (\varphi_1 \otimes \varphi_2)) = (\varphi_1 \otimes \varphi_2, (E^{H_1} \circledast E^{H_{2,\alpha}} )_\lambda (\varphi_1 \otimes \varphi_2))
\\
= \int_{\lambda_1 + \lambda_2 < \lambda} d (\varphi_1, E^{H_1}_{\lambda_1} \varphi_1) \, d (\varphi_2, E^{H_{2,\alpha}}_{\lambda_2} \varphi_2)
= \int_{-\infty}^\lambda (f_1*f_2)(\xi)\, d\xi \quad.
\end{eqnarray*}
As also convolution $f_1*f_2 \in L^1({\mathbb R}) \cap L^\infty({\mathbb R})$, $\varphi_1 \otimes \varphi_2 \in {\mathcal M}(H^{(+)})$ and $M\subset {\mathcal M}(H^{(+)})$ by the linearity.
\end{proof}
We are going to verify the assumptions of Theorem \ref{Kuroda-Birman} in the Appendix B with our $C_1$, $C_2$, $M$, $A=H$, $B=H^{(+)}$. Let us start with (\ref{C2decrease2}).
Taking into account that the difference of resolvents in $C_2$ is a bounded operator, $K$ commutes with $\exp(-i H^{(+)} t)$, and
$$
K\psi = \varphi \otimes \varphi_0 \; , \; \hat{\varphi} \in C_0^\infty (0,{\alpha}/{2}) \; , \; e^{-i H^{(+)}t} K\psi = \left(e^{-i H_1 t} \varphi \right) \otimes e^{i\frac{\alpha^2}{4}t} \varphi_0
$$
for every $\psi \in M$ with a suitable $\varphi \in L^2({\mathbb R})$ it is sufficient to verify the following statement.
\begin{lemma}
Let $\varphi \in L^2({\mathbb R})$, $\hat{\varphi} \in C_0^\infty ((0,\alpha/2))$ where $\alpha>0$ then
$$
\lim_{t \to +\infty}\|e^{-i H_1 t}\varphi\|_{L^2((-\infty,0))} = 0 \quad,\quad \int_0^{+\infty} \|e^{-i H_1 t}\varphi\|_{L^2((-\infty,0))} \, dt <\infty \; .
$$
\end{lemma}
\begin{proof} Let us use the explicit form of $\phi(t,x) = (\exp({-iH_1 t})\varphi)(x)$,
$$
({\mathcal F} e^{-i H_1 t} \varphi)(p) = e^{-i p^2 t} \hat{\varphi}(p) \quad,\quad \phi(t,x) = (2\pi)^{-\frac{1}{2}} \int_\varepsilon^a e^{ipx -i p^2 t} \hat{\varphi}(p) \, dp
$$
where $0<\varepsilon <a \leq \alpha/2$ are such that still $\hat{\varphi} \in C_0^\infty ((\varepsilon, a))$. Integrating here twice by parts with respect to p for $x<0$ and $t>0$,
\begin{eqnarray*}
\phi(t,x)=
\\
-(2\pi)^{-1/2} \int_\varepsilon ^a e^{ipx-ip^2t}\left(\frac{\hat{\varphi}''(p)}{(x-2pt)^2} + \frac{6\hat{\varphi}'(p)t}{(x-2pt)^3} + \frac{12\hat{\varphi}(p)t^2}{(x-2pt)^4} \right) \, dp
\end{eqnarray*}
and
$$
|\phi(t,x)| \leq c (x-2\varepsilon t)^{-2}
$$
with a constant $c$ dependent on $\varphi$. Then
$$
\|\phi(t,\cdot)\|_{L^2((-\infty,0))} \leq \frac{c}{2\sqrt{6}} \varepsilon^{-3/2} t^{-3/2}
$$
and the statement follows, integrability near $t=0$ being automatic as $\|\phi(t,\cdot)\|_{L^2({\mathbb R})}$ is independent of $t$.
\end{proof}

It remains to show that the operator $C_1$ is trace class. We typically deal with the integral operators $\hat{J}$ in $L^2({\mathbb R}^2)$ with the kernels of the form
$$
J(x,y)=\int_{\mathbb R} A(x,s) B(s,y) \, ds \quad,\quad (x,y \in {\mathbb R}^2) \quad.
$$
According to Proposition 3.6.5 of \cite{BSimon} the operator $\hat{J} \in {\mathcal J}_1$ if and only if
$$
\sup_{\{\varphi_n\},\{\psi_n\}} S(\hat{J}, \{\varphi_n\}, \{\psi_n\}) <\infty     \;\; ,\;\; S(\hat{J}, \{\varphi_n\}, \{\psi_n\}) =  \sum_n |(\varphi_n,\hat{J} \psi_n)|
$$
where the supremum is taken over all orthonormal sets $\{\varphi_n\}$, $\{\psi_n\}$.
We use the following scheme.
Using Bessel and Schwarz inequalities,
\begin{eqnarray}
\nonumber
S = \sum_n \left| \int_{{\mathbb R}^2} \overline{\varphi_n(x)} \left(\int_{\mathbb R} A(x,s) \left(\int_{{\mathbb R}^2} B(s,y) \psi_n(y)\, dy\right) \, ds\right)\,  dx \right|
\\
\nonumber
\leq \int_{\mathbb R} \sum_n | (\varphi_n,A(\cdot,s)) | | (\overline{B(s,\cdot)},\psi_n) | \, ds
\\
\nonumber
\leq \int_{\mathbb R} \left[ \sum_n |(\varphi_n,A(\cdot,s))|^2 \right]^{1/2} \left[ \sum_n |(\psi_n,B(s,\cdot))|^2 \right]^{1/2} \, ds
\\
\label{ABnorms}
\leq \int_{\mathbb R} \left[ \int_{{\mathbb R}^2}|A(x,s)|^2 \,dx \right]^{1/2} \left[ \int_{{\mathbb R}^2} |B(s,y)|^2 \,dy \right]^{1/2} \, ds \quad
\end{eqnarray}
and it is sufficient to verify the convergence of the last integral.
This also justifies the use of Fubini theorem here in the following way. The Fourier coefficient
$$
|(\overline{B(s,\cdot)}, \psi_n)| \leq \| B(s,\cdot)\| \quad.
$$
As functions $A$ and $B$ involved in our estimates below are measurable, it is sufficient to verify the convergence of
$$
\int_{{\mathbb R}^2 \times {\mathbb R}} \left| \overline{\varphi_n(x)} A(x,s) (\overline{B(s,\cdot)}, \psi_n) \right| \,dx\,ds \leq \int_{{\mathbb R}} \|A(\cdot,s)\|\,\|B(s,\cdot)\| \,ds \quad,
$$
i.e., the finiteness of (\ref{ABnorms}). The interchange of the sum  and integral over $s$ with the non-negative integrand is then always possible.
\begin{lemma}
The operator
$$
C_1 = \left( (H^{(+)}+\kappa^2)^{-1}-(H+\kappa^2)^{-1} \right) \theta(x_1) K
$$
is trace class for $\kappa$ sufficiently large.
\end{lemma}
\begin{proof}
Remember that by (\ref{spectproj})
$$
K=K_2 K_1 \quad,\quad K_1={\mathcal F}_1^{-1}\chi_{(0,\alpha/2)}{\mathcal F}_1 \quad,\quad K_2=\varphi_0 (\varphi_0,\cdot)_2 \quad.
$$
With the help of relations (\ref{full_resolvent}), (\ref{reduced_resolvent}) and denoting quantities related to the operator $H^{(+)}$ by superscript $(+)$
while leaving those related to $H$ without a superscript,
\begin{eqnarray*}
C_1&=&\alpha \sum_{j=1}^5 C_{1j}\, K_1 \quad,
\\
C_{11}&=&(R_{dx\,m}^{(+)}-R_{dx\,m})R_{m\,dx} \theta(x_1) K_2 \quad,
\\
C_{12}&=&(R_{dx\,m}^{(+)}-R_{dx\,m}) r R_{m\,dx} \theta(x_1) K_2 \quad,
\\
C_{13}&=&R_{dx\,m}^{(+)}(R_{m\,dx}^{(+)}-R_{m\,dx}) \theta(x_1) K_2 \quad,
\\
C_{14}&=&R_{dx\,m}^{(+)} (r^{(+)} - r) R_{m\,dx} \theta(x_1) K_2 \quad,
\\
C_{15}&=&R_{dx\,m}^{(+)} r^{(+)} (R_{m\,dx}^{(+)}-R_{m\,dx}) \theta(x_1) K_2 \quad.
\end{eqnarray*}
As $K_1$ is a bounded operator, it is sufficient to show that each term $C_{1j}$ is trace class separately using the above described scheme. We use the properties of Macdonald function and the explicit form of $K_2$, including that of eigenfunction $\varphi_0$. Again, $c$ denotes a constant whose value may change from line to line.
\\
\\
{\em Operator $C_{11}$}
\\
Put
\begin{eqnarray*}
A(x,s)=(2\pi)^{-1} \left(K_0(\kappa|x-\Gamma(s)|) - K_0(\kappa|x-\Gamma^{(+)}(s)|) \right) \quad,
\\
B(s,y)=(2\pi)^{-1} \int_{\mathbb R} K_0(\kappa|\Gamma(s)-(y_1,z_2)|) \theta(y_1) \varphi_0(z_2)\, dz_2 \varphi_0(y_2) \quad.
\end{eqnarray*}

For any $s$ we can estimate
\begin{equation}
\label{C111}
\int_{{\mathbb R}^2} |A(x,s)|^2 \, dx \leq 2 (2\pi)^{-2} \int_{{\mathbb R}^2} K_0(\kappa|x|)^2\,dx < \infty.
\end{equation}
Further,
$$
|B(s,y)| \leq (2\pi)^{-1} \left[\int_{\mathbb R} K_0(\kappa |\Gamma(s)-(y_1,z_2)|)^2 \, dz_2 \right]^{1/2} \varphi_0(y_2)
$$
and
\begin{equation}
\label{C112}
\int_{{\mathbb R}^2} |B(s,y)|^2\,dy \leq (2\pi)^{-2} \int_{{\mathbb R}^2} K_0(\kappa|y|)^2\,dy < \infty
\end{equation}
using normalization of $\varphi_0$.

Consider now the values of $s>0$. Then (see (\ref{nudef}) and (\ref{K0bound}))
\begin{eqnarray*}
|A(x,s)| = (2\pi)^{-1} \left|\kappa \int^{|x-\Gamma^{(+)}(s) - \nu_+(s)|}_{|x-\Gamma^{(+)}(s)|} K_1(\kappa t) \, dt \right|
\\
\leq
c \kappa \left| \int^{|x-\Gamma^{(+)}(s) - \nu_+(s)|}_{|x-\Gamma^{(+)}(s)|} \frac{1}{\kappa t} e^{-\frac{3}{4} \kappa t} \, dt \right|
\\
\leq
c e^{-\frac{3}{4} \kappa \min(|x-\Gamma^{(+)}(s)|,|x-\Gamma^{(+)}(s) - \nu_+(s)|)} \left|\log \frac{|x-\Gamma^{(+)}(s) - \nu_+(s)|}{|x-\Gamma^{(+)}(s)|} \right|
\end{eqnarray*}
and
$$
\int_{{\mathbb R}^2} |A(x,s)|^2 dx \leq c \int_{{\mathbb R}^2} e^{\frac{3}{2}\kappa |\nu_+(s)| - \frac{3}{2}\kappa|x|} \left|\log\frac{|x-\nu_+(s)|}{|x|} \right|^2 \, dx
\quad.
$$
Remembering that $\nu_+(s)$ is bounded for $s>0$,
\begin{eqnarray*}
\int_{|x|\leq 2|\nu_+(s)|} |A(x,s)|^2 \,dx
\\
\leq c \int_{|x|\leq 2|\nu_+(s)|} |x| \left( (\log||x|-|\nu_+(s)||)^2 + (\log||x|+|\nu_+(s)||)^2
\right.
\\
\left.
 + (\log|x|)^2 \right) \, d|x| \leq c |\nu_+(s)| \quad,
 \\
 \int_{|x|>2|\nu_+(s)|} |A(x,s)|^2 \,dx \leq c \int_{|x|>2|\nu_+(s)|} |x| e^{-\frac{3}{2}\kappa |x|} \frac{|\nu_+(s)|^2}{|x|^2} \, d|x| \leq c |\nu_+(s)|
\end{eqnarray*}
so
$$
\int_{{\mathbb R}^2} |A(x,s)|^2 dx \leq c |\nu_+(s)| \quad.
$$
Now
\begin{equation}
\label{C113}
\int_0^{+\infty} \|A(\cdot,s)\| \, ds < \infty
\end{equation}
under our assumptions according to Proposition \ref{curvelimits}.

For $s<0$ let us better estimate the term $B$. Assume first $v_- \not= v_+$ (i.e. $v_{-2}\not=0$ for our choice (\ref{specialv+})).
Let us estimate with the help of (\ref{K0bound})
\begin{eqnarray*}
\int_{\mathbb R} K_0(\kappa |y|)^2 \,dy_1 \leq c \int_{\mathbb R} (1+ |\log\kappa |y_1||^2 + \kappa^2 (y_1^2 + y_2^2)) e^{-\sqrt{2}\kappa (|y_1|+|y_2|)} \,dy_1
\\
\leq c (1+\kappa ^2 y_2^2) e^{-\sqrt{2}\kappa |y_2|} \quad.
\end{eqnarray*}
Then using Schwarz inequality and that $\varphi_0 \in L^1({\mathbb R})$ for the integral over $z_2$
\begin{eqnarray*}
|B(s,y)|^2 \leq c\, \theta(y_1) |\varphi_0(y_2)|^2 \int_{\mathbb R} |K_0(\kappa | \Gamma(s) - (y_1,z_2)|)|^2 \varphi_0(z_2) \, dz_2 \quad,
\\
\int_{{\mathbb R}^2} |B(s,y)|^2 \,dy \leq c \int_{\mathbb R} (1+\kappa^2 |\Gamma_2(s)-z_2|^2) e^{-\sqrt{2}\kappa |\Gamma_2(s)-z_2|} \, \varphi_0(z_2)\,dz_2
\\
\leq c \int_{\mathbb R} e^{-\kappa |\Gamma_2(s)-z_2| -\frac{\alpha}{2}|z_2|} \,dz_2 \leq c e^{-\min(\kappa,\alpha) |\Gamma_2(s)|/2}
\leq c e^{\min(\kappa,\alpha) |v_{-2}| s/2}
\end{eqnarray*}
as $|\Gamma_2(s)|\geq -|v_{-2}|s -|a_-|-\|\nu_-\|_\infty$ for $s<0$.
Now
\begin{equation}
\label{C114}
\int_{-\infty}^0 \|B(s,\cdot)\| \,ds < \infty \quad.
\end{equation}

It remains to treat the case $v_-=v_+=(1,0)$, $s<0$. Here,
\begin{eqnarray*}
|\Gamma(s)-(y_1,z_2)|=|(s-y_1,-z_2)+a_-+\nu_-(s)| \geq |s-y_1 + a_{-1} + \nu_{-1}(s)| ,
\\
|B(s,y)|\leq (2\pi)^{-1} \|\varphi_0\|_{L^1} \, \theta(y_1) \, \varphi_0(y_2) \, K_0(\kappa |s-y_1 + a_{-1} + \nu_{-1}(s)|) ,
\end{eqnarray*}
\begin{eqnarray*}
\int_{{\mathbb R}^2} |B(s,y)|^2\, dy \leq c \int_0^{+\infty} K_0(\kappa |s-y_1 + a_{-1} + \nu_{-1}(s)|)^2 \, dy_1
\\
\leq c \int_{-s - |a_-| - \|\nu_-\|_\infty}^{+\infty} K_0(\kappa |y_1|)^2 \, dy_1 ,
\end{eqnarray*}
\begin{eqnarray*}
\int_{-\infty}^0 \|B(s,\cdot)\| \, ds
\leq
c \int_0^{|a_-| + \|\nu_-\|_\infty} \left[\int_{\mathbb R} K_0(\kappa |y_1|)^2\, dy_1 \right]^{1/2} \, ds
\\
+ c \int_0^{+\infty}\left[\int_s^{+\infty} K_0(\kappa |y_1|)^2\, dy_1 \right]^{1/2} \, ds
\\
\leq
c + c \int_0^{+\infty} \left[\int_s^{+\infty} (1+ (\log \kappa y_1)^2) e^{-2 \kappa y_1} \right]^{1/2} \, ds
\\
\leq c + c \int_0^{+\infty} \left[e^{-\kappa s} \int_s^{+\infty} (1+ (\log \kappa y_1)^2) e^{-\kappa y_1}  \right]^{1/2} \, ds < \infty ,
\end{eqnarray*}
i.e., (\ref{C114}) holds for every $v_- \not= - v_+$. Combining estimates (\ref{C111})-(\ref{C114}), the integral (\ref{ABnorms}) is finite and $C_{11} \in {\mathcal J}_1$ is proved.
\\
\\
{\em Operator $C_{12}$}
\\
This is an integral operator with the kernel
$$
\int_{{\mathbb R}^2} A(x,s) {\mathcal R}(s,t) B(t,y) \, ds\, dt
$$
where $A$ and $B$ are the same as for $C_{11}$. Similarly as above,
$$
\sum_n |(\varphi_n, C_{12} \psi_n)| \leq \int_{{\mathbb R}^2} \|A(\cdot,s)\| \, {\mathcal K}(s-t) \, \|B(t,\cdot)\| \, ds \, dt
$$
for every orthonormal sets $\{\varphi_n\}$, $\{\psi_n\}$ with ${\mathcal K}$ introduced in (\ref{kernelK}).
Let us divide the integration range into four parts,
${\mathbb R}^2 = ({\mathbb R}_+\times {\mathbb R}_+) \cup ({\mathbb R}_+\times {\mathbb R}_-) \cup ({\mathbb R}_-\times {\mathbb R}_+) \cup ({\mathbb R}_-\times {\mathbb R}_-)$.
In view of (\ref{C111})-(\ref{C114}) and ${\mathcal K} \in L^1({\mathbb R})$, integrals over $({\mathbb R}_+\times {\mathbb R}_+) \cup ({\mathbb R}_+\times {\mathbb R}_-) \cup ({\mathbb R}_-\times {\mathbb R}_-)$
are finite. Further,
\begin{eqnarray*}
\int_{s<0\, , \, t>0} \|A(\cdot,s)\| \, |{\mathcal K}(s-t)| \, \|B(t,\cdot)\| \, ds \, dt
\\
\leq
c \int_{s<0\, , \, t>0} |{\mathcal K}(s-t)| \, ds \, dt
= \int_{-\infty}^0 |u| |{\mathcal K}(u)| \, du <\infty
\end{eqnarray*}
as ${\mathcal K} \in L^1({\mathbb R},(1+|u|)\, du)$ (see below (\ref{kernelK}) above). The relation $C_{12}\in {\mathcal J}_1$ is proved now.
\\
\\
{\em Operator $C_{13}$}
\\
To prove that $C_{13}$ is trace class we take
$$
A(x,s)=(2\pi)^{-1} K_0(\kappa |x-\Gamma^{(+)}(s)|) \quad,\quad \|A(\cdot,s)\| <c < \infty \quad,
$$
\begin{eqnarray*}
B(s,y)=(2\pi)^{-1} \theta(y_1) \overline{\varphi_0(y_2)} \cdot
\\
\cdot \int_{\mathbb R} \left(K_0(\kappa |\Gamma^{(+)}(s)-(y_1,z_2)|) - K_0(\kappa |\Gamma(s)-(y_1,z_2)|)\right) \varphi_0(z_2) \,dz_2
\; .
\end{eqnarray*}
Let us start with the estimate (see (\ref{K0bound}))
\begin{eqnarray*}
| K_0(\kappa |x|) - K_0(\kappa |y|) | = \left| \kappa \int_{|y|}^{|x|} K_1(\kappa t) \, dt\right| \leq c \kappa \left|\int_{|y|}^{|x|} (\kappa t)^{-1} e^{-\frac{3}{4} \kappa t} \, dt\right|
\\
\leq c e^{-\frac{3}{4} \kappa \min(|x|,|y|)} \left| \log \frac{|x|}{|y|} \right|
\end{eqnarray*}
and further
\begin{eqnarray*}
\int_{{\mathbb R}^2} |B(s,y)|^2 \, dy
\\
\leq c \int_{{\mathbb R}^2} \theta(y_1) |\varphi_0(y_2)|^2 e^{-\frac{3}{2} \kappa \min(|\Gamma^{(+)}(s)-y|, |\Gamma(s)-y)|) } \left| \log \frac{|\Gamma^{(+)}(s)-y|}{|\Gamma(s)-y)|}\right|^2 \, dy
\\
\leq c \int_{{\mathbb R}^2} e^{-\frac{3}{2} \kappa \min(|y|,|y-\nu_+(s)|)} \left| \log \frac{|y-\nu_+(s)|}{|y|} \right|^2 \, dy
\end{eqnarray*}

Consider now the case $s>0$ where $\|\nu_+\|_{L^\infty((0,+\infty))} < \infty$.
Realizing that
$$
|\log|y-\nu_+(s)||^2 \leq |\log||y|+|\nu_+(s)|||^2 + |\log||y|-|\nu_+(s)|||^2
$$
we estimate
$$
\int_{|y|\leq 2 |\nu_+(s)|} \left| \log \frac{|y-\nu_+(s)|}{|y|}\right|^2 \, dy \leq c |\nu_+(s)| \quad.
$$
For $|y|>2 |\nu_+(s)|$,
$$
|y|\left|\log \frac{|y-\nu_+(s)|}{|y|} \right|^2 \leq c |\nu_+(s)| \quad.
$$
Finally, we obtain
$$
\|B(s,\cdot)\| \leq c \sqrt{|\nu_+(s)|}
$$
which is integrable in ${\mathbb R}_+$ due to the Proposition \ref{curvelimits}.

For $s<0$, we apply the considerations used for $C_{11}$ to the both parts depending on $\Gamma$ and $\Gamma^{(+)}$ and we obtain
$$
\int_{\mathbb R} \|B(s, \cdot)\| \, ds < \infty
$$
which leads to the finiteness of (\ref{ABnorms}) and so proofs $C_{13}\in{\mathcal J}_1$.
\\
\\
{\em Operator $C_{14}$}
\\
We consider $\kappa$ large enough for which the identity
$$
r^{(+)} - r = \alpha (I+r) (R_{mm}^{(+)} - R_{mm}) (I+r^{(+)}) \quad.
$$
holds. Then we can write
\begin{eqnarray*}
C_{14}&=& \alpha (U_1+U_2+U_3+U_4) \theta(x_1) K_2 \quad,
\\
U_1&=& R_{dx\, m}^{(+)} (R_{mm}^{(+)} - R_{mm}) R_{m\,dx} \quad,
\\
U_2&=& R_{dx\, m}^{(+)} r (R_{mm}^{(+)} - R_{mm}) R_{m\,dx} \quad,
\\
U_3&=& R_{dx\, m}^{(+)} (R_{mm}^{(+)} - R_{mm}) r^{(+)} R_{m\,dx} \quad,
\\
U_4&=& R_{dx\, m}^{(+)} r (R_{mm}^{(+)} - R_{mm}) r^{(+)} R_{m\,dx} \quad.
\end{eqnarray*}
It is sufficient to show that the operators $U_1,\dots,U_4$ are trace class. They are of the form $U_j = R_{dx\, m}^{(+)} \hat{G}_j R_{m\,dx}$ where
$\hat{G}_j$ are integral operators with the kernels $G_j(s,t)$ in $L^2({\mathbb R})$ for $j=1,\dots,4$.
Let us denote for a while as $A(x,s)$ and $B(s,x)$ the kernels of operators $R_{dx\, m}^{(+)}$ and $R_{m\,dx}$ given in (\ref{Rxm}) and (\ref{Rmx}). As above, it is sufficient to show that
$$
\int_{{\mathbb R}^2} \|A(\cdot,s)\| |G_j(s,t)| \|B(t,\cdot)\| \, ds \, dt < \infty \quad.
$$
As $\|A(\cdot,s)\|$, $\|B(s,\cdot)\|$ are $s$-independent constants, it is sufficient to show that
$$
\int_{{\mathbb R}^2} |G_j(s,t)|  \, ds \, dt < \infty \quad.
$$
Function $G_1=-g$ (\ref{Rmm-R0kernel}) and using (\ref{Rmm-R0_HS})
\begin{eqnarray*}
\int_{{\mathbb R}^2} |G_1(s,t)| \,ds\,dt
\\
\leq \left[\int_{{\mathbb R}^2} w(s)^{-\delta} w(t)^{-\delta}\,ds\,dt\right]^{1/2} \left[\int_{{\mathbb R}^2} w(s)^\delta g(s,t) w(t)^\delta \,ds\,dt \right]^{1/2} < \infty
\end{eqnarray*}
where $\delta>3$ is a number from Assumption 1 (at this point $\delta>1$ is sufficient of course).
$$
G_2(s,t) = \int_{\mathbb R} {\mathcal R}(s,u) G_1(u,t) \,du
$$
and
$$
\int_{{\mathbb R}^2} |G_2(s,t)| \,ds\,dt \leq \int_{{\mathbb R}^3} |{\mathcal K}(s-u)| |G_1(u,t)| \,ds\,du\,dt < \infty
$$
as ${\mathcal K} \in L^1({\mathbb R})$.
Similarly are integrable
\begin{eqnarray*}
G_3(s,t)&=& \int_{\mathbb R} G_1(s,u) {\mathcal R}^{(+)}(u,t) \,du \quad,
\\
G_4(s,t)&=& \int_{{\mathbb R}^2} {\mathcal R}(s,u) G_1(u,q) {\mathcal R}^{(+)}(q,t) \,du \,dq
\end{eqnarray*}
and $C_{14}\in {\mathcal J}_1$ is proved.
\\
\\
{\em Operator $C_{15}$}
\\
Here we need to prove that
$$
\int_{{\mathbb R}^2}  \|A(\cdot,s)\| |{\mathcal R}^{(+)}(s,t)| \|B(t,\cdot)\| \,ds\,dt < \infty
$$
where $A$ and $B$ are the same as in the case of $C_{13}$. The estimates there together with (\ref{kernelK}) give the required relation.
This completes the proof of the Lemma.
\end{proof}
We formulate the result proved in this section.
\begin{theorem}
\label{woexistence}
Under the Assumptions 1-4, the wave operators $\Omega^{(+)}_{in}$, $\Omega^{(+)}_{out}$, $\Omega^{(-)}_{in}$, $\Omega^{(-)}_{out}$ defined in (\ref{wo_first})-(\ref{wo_last}) exist.
\end{theorem}
\section{Conclusions}
\label{conclusionsection}
Our results are contained in Theorems \ref{spectralac} and \ref{woexistence} above. For the quantum mechanical particles transversally bounded to the asymptotically flat curve
by the attractive contact $\delta$-interaction with the coupling constant $\alpha$,
with energies in $(-\alpha^2/4,0)$ (i.e. in the range accessible due to the coupling only) we show that
\begin{enumerate}
\item[(1)]
the spectrum in this range is absolutely continuous with possible embedded eigenvalues, discrete in every compact subinterval of $(-\alpha^2/4,0)$;
\item[(2)]
the wave operators defined above exist.
\end{enumerate}
Our proofs are done for $C^3$-smooth non-intersecting curve with rather severe asymptotic conditions for approaching to the mutually diverging asymptotic half-lines.

We should stress that we proved the existence of
$$\Omega^{(\pm)}_{in,out}:=\Omega_{in,out}(H,H^{(\pm)}; P^{(\pm)}_{\gtrless}J^{(\pm)})$$
only. We know nothing on the wave operators of the type
$$
\Omega_{in,out}(H^{(\pm)},H) \quad.
$$
This means that our wave operators map each scattering state of the "free" Hamiltonians $H^{(\pm)}$ (incoming or outgoing along the asymptotic half-lines)
to the one of $H$ but we did not prove that all scattering states of $H$ are covered and
the opposite mappings exist. The reason is the used simple explicit form of the projectors $P^{(\pm)}_{\gtrless}J^{(\pm)}$ while the form of spectral projector for $H$ is not known.
However, our result on the absence of singular continuous spectrum shows that if the wave operators would be complete they would be also asymptotically complete.

The other question desiring further study is the existence of the embedded eigenvalues and especially their accumulation near the points $-\alpha^2/4$ and $0$ which one perhaps would not expect
but which is not excluded by our proof.

\section{Appendices}
\appendix
\section{Assumptions of Exner and Ichinose}
To verify the validity of the assumptions of \cite{EI2001} under our assumptions, it is sufficient to check (a2) of \cite{EI2001}, i.e. to show the existence
of $d>0$, $\mu>0$, $\omega \in (0,1)$ such that for $\omega < \frac{s}{s'} < \omega^{-1}$
\begin{equation}
\label{a2_below}
\frac{|\Gamma(s)-\Gamma(s')|}{|s-s'|} \geq 1 - \frac{d}{\sqrt{1+|s+s'|^{2\mu}}} \quad.
\end{equation}
Here the left hand side extends to $1$ for $s=s'$ (see the proof of Proposition \ref{assump_equivalence}).
We shall assume that $s$ and $s'$ have the same sign as only such ones enter the assumption. Due to the symmetry in $s$ and $s'$ we may also assume
that $|s|\leq |s'|$. Let $R_1$ and $\kappa \in(1,\delta + \frac{1}{2})$ be the numbers from Proposition \ref{curvelimits}, we assume $R_1\geq 1$ enlarging its value
if needed. Choose arbitrary $\omega\in (0,1)$, $\mu=\kappa$,
$R_2\geq \omega^{-1}R_1 > R_1$.

Let us first assume $|s| \leq |s'| \leq R_2$.
The left hand side of (\ref{a2_below}) has a lower bound $\rho\in (0,1]$ by the relation (\ref{rhodef}) from Assumption 2.
The right hand side of (\ref{a2_below}) has a maximum
$$
1 - \frac{d}{\sqrt{1+2^{2\mu}R_2^{2\mu}}}
$$
in the considered range of $s$, $s'$ and (\ref{a2_below}) will be satisfied for
\begin{equation}
\label{d1estim}
d\geq (1-\rho) \sqrt{1+2^{2\mu}R_2^{2\mu}} \quad.
\end{equation}

Let us further assume $|s'|\geq R_2 > R_1$ , then also $|s|\geq \omega |s'| \geq R_1$. Using relations (\ref{nudef}) and (\ref{nufall}) for $s,s'\gtrless 0$
respectively,
\begin{eqnarray*}
\nonumber
\frac{|\Gamma(s)-\Gamma(s')|}{|s-s'|} = \frac{| v_\pm (s-s') + \nu_\pm(s)-\nu_\pm(s') |}{|s-s'|}
\\
\geq 1 - \sup_{s \lessgtr \xi \lessgtr s'} |\nu_\pm'(\xi)|
\geq 1 - |s|^{-\mu} \quad.
\end{eqnarray*}
On the other hand,
\begin{eqnarray*}
\frac{d}{\sqrt{1+|s+s'|^{2\mu}}} = \frac{d}{\sqrt{1+(|s|+|s'|)^{2\mu}}} \geq \frac{d}{\sqrt{1+2^{2\mu}|s'|^{2\mu}}}
\\
\geq \frac{d}{\sqrt{1+2^{2\mu}\omega^{-2\mu}|s|^{2\mu}}}
\geq  2^{-\mu-\frac{1}{2}} d \omega^{\mu} |s|^{-\mu}
\end{eqnarray*}
where the relation $2\omega^{-1}|s| \geq R_1 \geq 1$ was used. Now
\begin{equation}
\label{d2estim}
d\geq 2^{\mu+\frac{1}{2}}\omega^{-\mu}
\end{equation}
is sufficient for the validity of (\ref{a2_below}).
Choosing $d$ satisfying both estimates (\ref{d1estim}) and (\ref{d2estim}) the assumption (a2) of \cite{EI2001} is verified.
\section{Generalized Pearson and Kuroda-Birman theorems}
We need a generalization of the Kuroda-Birman theorem (Theorem XI.9 of \cite{RSIV}, vol. 3, or more general Proposition 4 in Chapter 16 of \cite{Baumg_Woll}).
We sketch the proof closely following that of \cite{RSIV}.

Let us remind the following definition: Let B be a self-adjoint operator in a Hilbert space ${\mathcal H}$ and $E_\lambda$ its spectral projectors. We denote as ${\mathcal M}(B)$ the set of all
$\varphi \in {\mathcal H}$ such that $d (\varphi, E_\lambda \varphi) = |f(\lambda)|^2 d \lambda$ with $f \in L^\infty({\mathbb R})$.

We start with the generalization of Pearson theorem (Theorem XI.7 of \cite{RSIV}).
\begin{theorem}
\label{Pearson}
Let $A$, $B$ be self adjoint operators in a Hilbert space ${\mathcal H}$, $J$ a bounded operator in ${\mathcal H}$,
$$
C=AJ-JB
$$
in the sense that for every $\varphi\in{\mathcal D}(A)$, $\psi\in{\mathcal D}(B)$
\begin{equation}
\label{C_form_sense}
(\varphi,C\psi)=(A\varphi,J\psi)-(\varphi,J B \psi) \quad.
\end{equation}
Let $C=C_1+C_2$ where $C_1 \in {\mathcal J}_1$ is a trace class operator and $C_2$ is bounded. Let there exists a dense subset $M$ of ${\mathcal M}(B)$ such that
\begin{equation}
\label{C2decrease}
\lim_{t\to + \infty} C_2 e^{-iBt} \varphi =0 \quad , \quad \int_0^{+\infty} \| C_2 e^{-iBt} \varphi \| dt <\infty
\end{equation}
for every $\varphi \in M$.
Then there exists
$$
\Omega_{out} (A,B;J) = s-\lim_{t \to + \infty} e^{iAt} J e^{-iBt} P_{ac}(B) \quad.
$$
\end{theorem}
\begin{proof}
We describe only amendments needed to the proof in \cite{RSIV}. Let us denote $W(t)=e^{iAt}Je^{-iBt}$. Then
$$
W(t)-W(s) = i \int_s^t e^{i A \mu} C_1 e^{-i B \mu} d\mu + i \int_s^t e^{i A \mu} C_2 e^{-i B \mu} d\mu
$$
in the form sense (\ref{C_form_sense}) and for $\varphi \in M$,
\begin{eqnarray*}
(W(t)-W(s))e^{-iBa}\varphi = q_1(s,t,a) + q_2(s,t,a) \quad ,
\\
q_j(s,t,a) = i\int_s^t e^{i A \mu} C_j e^{-i B (\mu + a)} \varphi \, d\mu  \quad (j=1,2) \quad.
\end{eqnarray*}
Limits $a \to +\infty$ with fixed $s$, $t$ of the both integrals can be taken under the integral sign as there are constant majorants $\|C_j\| \|\varphi\|$ ($j=1,2$). The
both limits are zero (for $j=1$ by the Lemma above Theorem XI.7 of \cite{RSIV}, for $j=2$ by the assumption). Thus Equation (20) in the proof of Theorem XI.7 of \cite{RSIV} holds in our case.
The part containing trace class operator $C_1$ can be treated in the same way as in \cite{RSIV}. The part containing $C_2$ is bounded by a sum of several terms of the form ($X$ is an operator
uniformly bounded in $s$ and $t$)
$$
|(\varphi, \int_0^a e^{iBu} e^{iBt} X C_2 e^{-iBs} e^{-iBu} du \, \varphi)| \leq \|\varphi\| \|X\| \int_s^{s+a} \|C_2 e^{-iBu}\varphi \|\, du
$$
or
$$
|(\varphi, \int_0^a e^{iBu} e^{iBt} C_2^* X  e^{-iBs} e^{-iBu} du \, \varphi)| \leq \|\varphi\| \|X\| \int_t^{t+a} \|C_2 e^{-iBu}\|\, du
$$
which tend to zero as $a \to + \infty$ and then $t,s \to + \infty$ due to the second assumption (\ref{C2decrease}). Finally, we obtain
$$
\lim_{s,t \to +\infty} \| (W(t) - W(s)) \varphi \|^2 =0
$$
for $\varphi \in M \subset P_{ac}(B) {\mathcal H}$. As $M$ is dense in $P_{ac}(B) {\mathcal H}$ the existence of $\Omega_{out} (A,B;J)$ follows.
\end{proof}

Now the generalization of the Kuroda-Birman theorem can be proved.
\begin{theorem}
\label{Kuroda-Birman}
Let $A$ and $B$ be self-adjoint operators in a Hilbert space ${\mathcal H}$, $K$ a bounded operator in ${\mathcal H}$, $z$ a complex number in the resolvent sets of both $A$ and $B$,
\begin{equation}
\label{Cresolvent}
K(B-z)^{-1} - (A-z)^{-1}K = C_1 + C_2
\end{equation}
where $C_1 \in {\mathcal J}_1({\mathcal H})$ is a trace class operator and $C_2$ is a bounded operator in ${\mathcal H}$. Let there further exists  a dense set $M\subset {\mathcal M}(B)$
such that
\begin{equation}
\label{C2decrease2}
\lim_{t\to +\infty} C_2 e^{-iBt} \varphi = 0 \quad,\quad \int_0^{+\infty} \| C_2 e^{-iBt} \varphi \| \, dt < \infty
\end{equation}
for every $\varphi \in M$. Then there exists
$$
\Omega_{out} (A,B;K) = s-\lim_{t \to + \infty} e^{iAt} K e^{-iBt} P_{ac}(B) \quad.
$$
\end{theorem}
\begin{proof}
It is a very slight amendment of that for Theorem XI.9 of \cite{RSIV}. Denote
$$
J=(A-z)^{-1} K (B-z)^{-1} \quad.
$$
Then $AJ-JB = C_1+C_2$ in the form sense and
\begin{eqnarray*}
s-\lim_{t\to +\infty} e^{iAt}Je^{-iBt} P_{ac}(B) =
\\
s-\lim_{t\to +\infty} a^{iAt} (A-z)^{-1} K (B-z)^{-1} e^{-iBt} P_{ac}(B)
\end{eqnarray*}
exists by Theorem \ref{Pearson}. Applying that to $(B-z)\varphi$, $\varphi \in {\mathcal D}(B)$ the existence of
\begin{eqnarray*}
\lim_{t\to +\infty} e^{iAt} (A-z)^{-1} K e^{-iBt}P_{ab}(B) \varphi =
\\
\lim_{t\to +\infty} e^{iAt} ( K (B-z)^{-1} - C_1 -C_2) e^{-iBt} P_{ac}(B) \varphi
\end{eqnarray*}
follows.
The terms with $C_1$, $C_2$ here tends to zero according to the assumptions and Lemma 2 above Theorem XI.7 in \cite{RSIV} for $\varphi \in M$.
As $M$ is dense in $P_{ac}(B){\mathcal H}$,
$$
s-\lim_{t\to +\infty} e^{iAt} K e^{-iBt} P_{ac} (B-z)^{-1}
$$
exists and then also $\Omega_{out}(A,B;K)$ exists because ${\mathcal D}(B)$ is dense in ${\mathcal H}$.
\end{proof}
%
%
\vspace{1cm}

\noindent
{\bf Acknowledgements} \\
The author is indebted to P. Exner, who also called his attention to the problem, J. Behrndt and V. Lotoreichik for discussions. The work was supported by Czech Science Foundation
project No.\ 17-01706S and NPI CAS institutional support RVO 61389005.


%
\end{document}